\documentclass[sn-basic,Numbered]{sn-jnl}% Basic Springer Nature Reference Style/Chemistry Reference Style
\usepackage{graphicx}%
\usepackage{multirow}%
\usepackage{amsmath,amssymb,amsfonts}%
\usepackage{amsthm}%
\usepackage{mathrsfs}%
\usepackage[title]{appendix}%
\usepackage{xcolor}%
\usepackage{textcomp}%
\usepackage{manyfoot}%
\usepackage{booktabs}%
\usepackage{algorithmicx}%
\usepackage{listings}%
\usepackage{cleveref}

\newcommand{\dd}{\mathinner{.\,.}}

\usepackage{tikz}
\usetikzlibrary{matrix,backgrounds,positioning, arrows}

\usepackage{algpseudocode}
\usepackage[ruled]{algorithm}
\algrenewcommand\algorithmicrequire{\textbf{Input:}}
\algrenewcommand\algorithmicensure{\textbf{Output:}}
\algtext*{EndWhile}% Remove "end while" text
\algtext*{EndIf}% Remove "end if" text
\algtext*{EndFor}% Remove "end if" text
\algtext*{EndFunction}% Remove "end if" text

\newcommand{\floor}[1]{\lfloor #1 \rfloor}

\newcommand{\deltalog}{\delta\log\frac{n\log\sigma}{\delta\log n}}

\newcommand{\gexp}{\mathtt{exp}}

\newcommand{\bwt}{\mathtt{bwt}}

\newcommand{\pred}{\texttt{pred}}
\newcommand{\suc}{\texttt{succ}}

\newcommand{\syma}{\texttt{a}}
\newcommand{\symb}{\texttt{b}}
\newcommand{\symc}{\texttt{c}}
\newcommand{\dol}{\texttt{\$}}
\newcommand{\one}{\texttt{1}}
\newcommand{\zero}{\texttt{0}}

\newcommand{\bo}{\mathcal{O}}
\newcommand{\rmq}{\texttt{rmq}}
\renewcommand{\time}[1]{t_{#1}}

\newcommand{\nsv}{\texttt{nsv}}
\renewcommand{\exp}[1]{\mathtt{exp}(#1)}
\newcommand{\RMQ}{\textsc{rmq}}
\newcommand{\PSV}{\textsc{psv}}
\newcommand{\NSV}{\textsc{nsv}}

\newcommand{\no}[1]{}

\theoremstyle{thmstyleone}%
\newtheorem{theorem}{Theorem}%  meant for continuous numbers
\newtheorem{proposition}[theorem]{Proposition}% 
\newtheorem{lemma}[theorem]{Lemma}% 
\newtheorem{corollary}[theorem]{Corollary}% 

\theoremstyle{thmstyletwo}%
\newtheorem{example}{Example}%

\theoremstyle{thmstylethree}%
\newtheorem{definition}{Definition}%

\begin{document}

\title[Generalized Straight-Line Programs]{Generalized Straight-Line Programs}


\author*[1,2]{\fnm{Gonzalo} \sur{Navarro}\email{gnavarro@uchile.cl}}

\author*[1,2]{\fnm{Francisco} \sur{Olivares}\email{folivares@uchile.cl}}

\author*[1,2]{\fnm{Cristian} \sur{Urbina}}\email{crurbina@dcc.uchile.cl}

\affil[1]{\orgdiv{Department of Computer Science}, \orgname{University of Chile}, \country{Chile}}

\affil[2]{\orgname{CeBiB --- Centre for Biotechnology and Bioengineering}, \country{Chile}}


\abstract{It was recently proved that any Straight-Line Program (SLP) generating a given string can be transformed in linear time into an equivalent balanced SLP of the same asymptotic size. We generalize this proof to a general class of grammars we call Generalized SLPs (GSLPs), which allow rules of the form $A \rightarrow x$ where $x$ is any Turing-complete representation (of size $|x|$) of a sequence of symbols (potentially much longer than $|x|$). We then specialize GSLPs to so-called Iterated SLPs (ISLPs), which allow rules of the form $A \rightarrow \Pi_{i=k_1}^{k_2} B_1^{i^{c_1}}\cdots B_t^{i^{c_t}}$ of size $2t+2$. We prove that ISLPs break, for some text families, the measure $\delta$ based on substring complexity, a lower bound for most measures and compressors exploiting repetitiveness. Further, ISLPs can extract any substring of length $\lambda$, from the represented text $T[1\dd n]$, in time $\bo(\lambda + \log^2 n\log\log n)$. This is the first compressed representation for repetitive texts breaking $\delta$ while, at the same time, supporting direct access to arbitrary text symbols in polylogarithmic time. We also show how to compute some substring queries, like range minima and next/previous smaller value, in time $\bo(\log^2 n \log\log n)$. Finally, we further specialize the grammars to Run-Length SLPs (RLSLPs), which restrict the rules allowed by ISLPs to the form $A \rightarrow B^t$. Apart from inheriting all the previous results with the term $\log^2 n \log\log n$ reduced to the near-optimal $\log n$, we show that RLSLPs can exploit balance to efficiently compute a wide class of substring queries we call ``composable''---i.e., $f(X \cdot Y)$ can be obtained from $f(X)$ and $f(Y)$. As an example, we show how to compute Karp-Rabin fingerprints of texts substrings in $\bo(\log n)$ time. While the results on RLSLPs were already known, ours are much simpler and require little precomputation time and extra data associated with the grammar.
}




\maketitle

\section{Introduction}

Motivated by the data deluge, and by the observed phenomenon that many of the fastest-growing text collections are highly repetitive, recent years have witnessed an increasing interest in (1) defining measures of compressibility that are useful for highly repetitive texts, (2) develop compressed text representations whose size can be bounded in terms of those measures, and (3) provide efficient (i.e., polylogarithmic time) access methods to those compressed texts, so that algorithms can be run on them without ever decompressing the texts \cite{Navacmcs20.3,Navacmcs20.2}. We call {\em lower-bounding measures} those satisfying (1), {\em reachable measures} those (asymptotically) reached by the size of a compressed representation (2), and {\em accessible measures} those reached by the size of representations satisfying (3).

For example, the size $\gamma$ of the smallest ``string attractor'' of a text $T$ is a lower-bounding measure, unknown to be reachable \cite{KP18}, and smaller than the size reached by known compressors. The size $b$ of the smallest ``bidirectional macro scheme'' of $T$ \cite{SS82}, and the size $z$ of the Lempel-Ziv parse of $T$ \cite{LZ76}, are reachable measures. The size $g$ of the smallest context-free grammar generating (only) $T$ \cite{CLLPPSS05} is an accessible measure \cite{BLRSRW15}. It holds $\gamma \le b \le z \le g$ for every text.

One of the most attractive lower-bounding measures devised so far is $\delta$ \cite{RRRS13,CEKNP20}. Let $T[1\dd n]$ be a text over alphabet $[1\dd\sigma]$, and $T_k$ be the number of distinct substrings of length $k$ in $T$, which define its so-called substring complexity. Then the measure is $\delta(T) = \max_k T_k/k$. This measure has several attractive properties: it can be computed in linear time and lower-bounds all previous measures of compressibility, including $\gamma$, for every text. While $\delta$ is known to be unreachable, the measure $\delta'=\deltalog$ has all the desired properties: $\Omega(\delta')$ is the space needed to represent some text family for each $n$, $\sigma$, and $\delta$; within $\bo(\delta')$ space it is possible to represent every text $T$ and access any length-$\lambda$ substring of $T$ in time $\bo(\lambda+\log n)$ \cite{KNP22}, together with more powerful operations \cite{KNP22,KNO23,KK23}. 

As for $g$, a {\em straight-line program (SLP)} is a context-free grammar that generates (only) $T$, and has size-2 rules of the form $A \rightarrow BC$, where $B$ and $C$ are nonterminals, and size-1 rules $A \rightarrow a$, where $a$ is a terminal symbol. The SLP size is the sum of all its rule sizes. A {\em run-length SLP (RLSLP)} may contain, in addition, size-2 rules of the form $A \rightarrow B^t$, representing $t$ repetitions of nonterminal $B$. A RLSLP of size $g_{rl}$ can be represented in $\bo(g_{rl})$ space, and within that space we can offer fast string access and other operations
%, fingerprint computation, and pattern searches 
\cite[App.~A]{CEKNP20}. It holds $\delta \le g_{rl} = \bo(\delta')$, where $g_{rl}$ is the smallest RLSLP that generates $T$ \cite{Navacmcs20.3,KNP22} (the size $g$ of the smallest grammar or SLP, instead, is not always $\bo(\delta')$). 

While $\delta$ lower-bounds all previous measures on every text, $\delta'$ is not the smallest accessible measure. 
In particular, $g_{rl}$ is always $\bo(\delta')$, and it can be smaller by up to a logarithmic factor. Indeed, $g_{rl}$ is a minimal accessible measure as far as we know. It is asymptotically between $z$ and $g$ \cite{Navacmcs20.3}. An incomparable accessible measure is $z_{end} \ge z$, the size of the LZ-End parse of the text \cite{KN13,KS22}.

The belief that $\delta$ is a lower bound to every reachable measure was disproved by the recently proposed L-systems \cite{NU21,NU23}. L-systems are like SLPs where all the symbols are nonterminals and the derivation ends at a specified depth in the derivation tree. The size $\ell$ of the smallest L-system generating $T[1\dd n]$ is a reachable measure of repetitiveness and was shown to be as small as $\bo(\delta/\sqrt{n})$ on some text families, thereby sharply breaking $\delta$ as a lower bound. Measure $\ell$, however, is unknown to be accessible, and thus one may wonder whether there exist accessible text representations that are smaller than $\delta$.

In this paper we present several contributions to this state of affairs:
\begin{enumerate}
\item We extend the result of Ganardi et al.~\cite{GJL2021}, which shows that any SLP of size $g$ generating a text of length $n$ can be balanced to produce another SLP of size $\bo(g)$ whose derivation tree is of height $\bo(\log n)$. Our extension is called {\em Generalized SLPs (GSLPs)}, which allow rules of the form $A \rightarrow x$ (of size $|x|$), where $x$ is a {\em program} (in any Turing-complete formalism) that outputs the right-hand side of the rule. We show that, if every nonterminal appearing in $x$'s output occurs at least twice, then the GSLP can be balanced in the same way as SLPs.
\item We explore a particular case of GSLP we call {\em Iterated SLPs (ISLPs)}. ISLPs extend SLPs (and RLSLPs) by allowing a more complex version of the rule $A \rightarrow B^t$, namely $A \rightarrow \Pi_{i=k_1}^{k_2} B_1^{i^{c_1}}\cdots B_t^{i^{c_t}}$, of size $2+2t$. We show that some text families are generated by an ISLP of size $\bo(\delta/\sqrt{n})$, thereby sharply breaking the $\Omega(\delta)$ barrier.
\item Using the fact that ISLPs are GSLPs and thus can be balanced, we show how to extract a substring of length $\lambda$ from the ISLP in time $\bo(\lambda+\log^2 n \log\log n)$, as well as computing substring queries like range minimum and next/previous smaller value, in time $\bo(\log^2 n \log\log n)$. ISLPs are thus the first accessible representation that can reach size $o(\delta)$.
\item Finally, we apply the balancing result to RLSLPs, which allow rules of the form $A \rightarrow B^t$. While the results on ISLPs are directly inherited (because RLSLPs are ISLPs) with the polylogs becoming just the nearly-optimal $\bo(\log n)$ \cite{VY13}, we give a general technique to compute a wide family of ``composable'' queries $f$ on substrings (i.e., $f(X \cdot Y)$ can be computed from $f(X)$ and $f(Y)$). As an application, we show how to compute Karp-Rabin fingerprints on text substrings in time $\bo(\log n)$, which we do not know how to do efficiently on ISLPs. This considerably simplifies and extends previous results \cite[App.~A]{CEKNP20}, as balanced grammars enable simpler algorithms that do not require large and complex additional structures. 
\end{enumerate}

This work is an extended version of articles published in SPIRE 2022 \cite{NOU2022} and LATIN 2024 \cite{NU2024}, which are now integrated into a coherent framework where specialized results are derived from more general ones, new operations are supported, and proofs are complete.

\section{Preliminaries}

We explain some concepts and notation used in the rest of the paper.

\subsection{Strings}

Let $\Sigma = [1\dd \sigma]$ be an ordered set of symbols called the \emph{alphabet}. A \emph{string} $T[1\dd n]$ of \emph{length} $n$ is a finite sequence $T[1]\,T[2] \dots T[n]$ of $n$ symbols in $\Sigma$. We denote by $\varepsilon$ the unique string of length $0$. We denote by $\Sigma^*$ the set of all finite strings with symbols in $\Sigma$. The $i$-th symbol of $T$ is denoted by $T[i]$; the notation $T[i\dd j]$ stands for the sequence $T[i]\dots T[j]$ for $1\le i \le j \le n$, and $\varepsilon$ otherwise. The \emph{concatenation} of $X[1\dd n]$ and $Y[1\dd m]$ is defined as $X \cdot Y = X[1] \dots X[n]\, Y[1] \dots Y[m]$ (we omit the dot when there is no ambiguity). If $T = XYZ$, then $X$ (resp. $Y$, resp. $Z$) is 
 a \emph{prefix} (resp. \emph{substring}, resp. \emph{suffix}) of $T$. A \emph{power} $T^k$ stands for $k$ consecutive concatenations of the string $T$.
 We denote by $|T|_a$ the number of occurrences of the symbol $a$ in $T$. A \emph{string morphism} is a function $\varphi: \Sigma^* \rightarrow \Sigma^*$ such that $\varphi(xy) = \varphi(x) \cdot \varphi(y)$ for any strings $x$ and $y$.

\subsection{Straight-Line Programs}

A \emph{straight-line program} (SLP) is a context-free grammar \cite{Sipser2012} that contains only terminal rules of the form $A \rightarrow a$ with $a \in \Sigma$, and binary rules of the form $A \rightarrow BC$ for variables $B$ and $C$ whose derivations cannot reach again $A$. These restrictions ensure that each variable of the SLP generates a unique string, defined as $\gexp(A) = a$ for a rule $A \rightarrow a$, and as $\gexp(A) = \gexp(B)\cdot\gexp(C)$ for a rule $A \rightarrow BC$. A \emph{run-length straight-line program} (RLSLP) is an SLP that also admits run-length rules of the form $A \rightarrow B^t$ for some $t \geq 3$, with their expansion defined as $\gexp(A)=\gexp(B)^t$. The {\em size} of an SLP $G$, denoted $|G|$, is the sum of the lengths of the right-hand sides of its rules; the size of an RLSLP is defined similary, assuming that rules $A \rightarrow B^t$ are of size 2 (i.e., two integers to represent $B$ and $t$).

The \emph{derivation or parse tree} of an SLP is an ordinal tree where the nodes are the variables, the root is the initial variable, and the leaves are the terminal variables. The children of a node are the variables appearing in the right-hand side of its rule (in left-to-right order). The \emph{height} of an SLP is the length of the longest path from the root to a leaf node in the derivation tree. The height of an RLSLP is obtained by \emph{unfolding} its run-length rules, that is, writing a rule $B^t$ as $\underbrace{B \dots B}_{t \textrm{ times}}$, to obtain an equivalent SLP (actually, a slight extension where the right-hand sides can feature more than two variables). 

The {\em grammar tree} is obtained by pruning the parse tree so that only the leftmost occurrence of a nonterminal is retained as an internal node and all the others become leaves. Rules $A \rightarrow B^t$ are represented as the node $A$ having a left child $B$ (which can be internal or a leaf) and a special right child denoting $B^{t-1}$ (which is a leaf). It is easy to see that the grammar tree has exactly $|G|+1$ nodes.

SLPs and RLSLPs yield measures of repetitiveness $g$ and $g_{rl}$, defined as the size of the smallest SLP and RLSLP generating the text, respectively. Clearly, it holds that $g_{rl} \le g$. It also has been proven that both $g$ and $g_{rl}$ are NP-hard to compute \cite{CLLPPSS05,KIKB24}.

\subsection{Other Repetitiveness Measures}

For self-containedness, we describe the most important repetitiveness measures and relate them with the accessible measures $g$ and $g_{rl}$; for more details see a survey \cite{Navacmcs20.3}.

\paragraph{Burrows-Wheeler Transform.}

The \emph{Burrows-Wheeler Transform} (BWT) \cite{BW94} is a reversible permutation of $T$, which we denote by $\bwt(T)$. It is obtained by sorting lexicographically all the rotations of the string $T$ and concatenating their last symbols, which can be done in $\bo(n)$ time. The measure $r$ is defined as the size of the \emph{run-length encoding} of $\bwt(T)$. Usually, $T$ is assumed to be appended with a sentinel symbol $\dol$ strictly smaller than any other symbol in $T$, and then we call $r_\dol$ the size of the run-length encoding of $\bwt(T\dol)$. This measure is then reachable, and fully-functional indexes of size $\bo(r_\dol)$ exist \cite{GNP18}, but interestingly, it is unknown to be accessible. While this measure is generally larger than others, it can be upper-bounded by $r_\dol = \bo(\delta\log\delta\log\frac{n}{\delta})$ \cite{KK20}.

\paragraph{Lempel-Ziv Parsing.}

The \emph{Lempel-Ziv parsing} (LZ) \cite{LZ76} of a text $T[1\dd n]$ is a \emph{factorization} into non-empty \emph{phrases}  $T = X_1X_2\dots X_z$ where each $X_i$ is either the first occurrence of a symbol or the longest prefix of $X_i \dots X_z$ with a copy in $T$ starting at a position in  $[1\dd |X_1\dots X_{i-1}|]$. LZ is called a \emph{left-to-right} parsing because each phrase has its \emph{source} starting to the left, and it is optimal among all parsings satisfying this condition. It can be constructed greedily from left to right in $\bo(n)$ time. The measure $z$ is defined as the number of phrases in the LZ parsing of the text, and it has been proved that $z \le g_{rl}$ \cite{NOPtit20}. While $z$ is obviously reachable, it is unknown to be accessible. A close variant $z_{end} \ge z$ \cite{KN13} that forces phrase sources to be end-aligned with a preceding phrase, has been shown to be accessible \cite{KS22}. 

\paragraph{Bidirectional Macro Schemes.}

A \emph{bidirectional macro scheme} (BMS) \cite{SS82} is a factorization of a text $T[1\dd n]$ where each phrase can have its source starting either to the left or to the right. The only requeriment is that by following the pointers from phrases to sources, we should eventually be able to fully decode the text. The measure $b$ is defined as the size of the smallest BMS representing the text. Clearly, $b$ is reachable, but it is unknown to be accessible. It holds that $b \le z$, and it was proved that $b \le r_\$$ \cite{NOPtit20}. Computing $b$ is NP-hard \cite{Gallant1982}. 

\paragraph{String Attractors.}

A \emph{string attractor} for a text $T[1\dd n]$ is a set of positions $\Gamma \subseteq [1\dd n]$ such that any substring of $T[i\dd j]$ has an occurrence $T[i'\dd j']$ crossing at least one of the positions in $\Gamma$ (i.e., there exist $k \in \Gamma$ such that $i'\le k\le j'$). The measure $\gamma$ is defined as the size of the smallest string attractor for the string $T$, and it is NP-hard to compute \cite{KP18}. It holds that $\gamma$ lower bounds the size $b$ of the smallest bidirectional macro scheme and can sometimes be asymptotically smaller \cite{BFIKMN21}. On the other hand, it is unknown if $\gamma$ is reachable.

\paragraph{Substring Complexity.}

Let $T[1\dd n]$ be a text and $T_k$ be the number of distinct substrings of length $k$ in $T$, which define its so-called substring complexity. Then the measure is $\delta = \max_k T_k/k$ \cite{RRRS13,CEKNP20}. This measure can be computed in $\bo(n)$  time and lower-bounds $\gamma$, and thus all previous measures of compressibility, for every text. On the other hand, it is known to be unreachable \cite{KNP22}. The related measure $\delta'=\deltalog$ is reachable and accessible, and still lower-bounds $b$ and all other reachable measures on some text family for every $n$, $\sigma$, and $\delta$ \cite{KNP22}. Besides, $g_{rl}$ (and thus $z$, $b$, and $\gamma$, but not $g$) are upper-bounded by $\bo(\deltalog)$; $g$ can be upper-bounded by $\bo(\gamma\log^2\frac{n}{\gamma})$ \cite{KNP22,KP18}.

\paragraph{L-systems.}

An \emph{L-system} (for compression) is a tuple $L = (V, \varphi, \tau, S, d, n)$ extending a traditional Lindenmayer system \cite{Lindenmayer1968a,Lindenmayer1968b}, where $V$ is the set of variables (which are also considered as terminal symbols), $\varphi: V \rightarrow V^+$ is the set of rules (and also a morphism of strings), $\tau: V \rightarrow V$ is a coding, $S \in V$ the initial variable, and $d$ and $n$ are integers. The string generated by the system is $\tau(\varphi^d(S))[1\dd n]$. The measure $\ell$ is defined as the size of the smallest L-system generating the string. It has been proven that $\ell$ is incomparable to $\delta$ ($\ell$ can be smaller by a $\sqrt{n}$ factor) and almost any other repetitiveness measure considered in the literature \cite{NU21,NU23}.

\section{Generalized SLPs and How to Balance Them} \label{sec:bal}

We introduce a new class of SLP which we show can be balanced so that its derivation tree is of height $\bo(\log n)$.

\bigskip\begin{definition} A \emph{generalized straight-line program} (GSLP) is an SLP that allows special rules of the form $A \rightarrow x$, where $x$ is a {\em program} (in any Turing-complete language) of length $|x|$ whose output $\mathtt{OUT}(x)$ is a nonempty sequence of variables, none of which can reach $A$. The rule $A \rightarrow x$ contributes $|x|$ to the size of the GSLP; the standard SLP rules contribute as usual. If it holds for all special rules that no variable appears exactly once inside $\mathtt{OUT}(x)$, then the GSLP is said to be \emph{balanceable}.
\end{definition}\bigskip

We can choose any desired language to describe the programs $x$. Though in principle $|x|$ can be taken as the Kolmogorov complexity of $\mathtt{OUT}(x)$, we will focus on very simple programs and on the asymptotic value of $x$.

We will prove that any balanceable GSLP can be balanced without increasing its asymptotic size. Our proof generalizes that of Ganardi et al.~\cite[Thm.~1.2]{GJL2021} for SLPs in a similar way to how it was extended to balance RLSLPs \cite{NOU2022}. Just as Ganardi et al., in this section we will allow SLPs to have rules of the form $A \rightarrow B_1 \cdots B_t$, of size $t$, where each $B_j$ is a terminal or a nonterminal; this can be converted into a strict SLP of the same asymptotic size.

We introduce now some definitions and state some results, from the work of Ganardi et al. \cite{GJL2021}, that we need in order to prove our balancing result for GSLPs.

A \emph{directed acyclic graph} (DAG) is a directed multigraph $D = (V,E)$ without cycles (nor loops). We denote by $|D|$ the number of edges in this DAG. For our purposes, we assume that any DAG has a distinguished node $r$ called the \emph{root}, satisfying that any other node can be reached from $r$ and $r$ has no incoming edges. We also assume that if a node has $k$ outgoing edges, they are numbered from $1$ to $k$, so edges are of the form $(u,i,v)$.  The \emph{sink nodes} of a DAG are the nodes without outgoing edges. The set of sink nodes of $D$ is denoted by $W$. We denote the number of paths from $u$ to $v$ as $\pi(u, v)$, and $\pi(u, V) = \sum_{v \in V}\pi(u, v)$ for a set $V$ of nodes. The number of paths from the root to the sink nodes is $n(D) = \pi(r, W)$.

One can interpret an SLP $G$ generating a string $T$ as a DAG $D$: There is a node for each variable in the SLP, the root node is the initial variable, variables of the form $A \rightarrow a$ are the sink nodes, and a variable with rule $A \rightarrow B_1B_2\dots B_t$ has outgoing edges $(A, i, B_i)$ for $i \in [1\dd t]$. Note that if $D$ is a DAG representing $G$, then $n(D) = |\gexp(G)| = |T|$.

\bigskip\begin{definition}{(Ganardi et al. \cite[page 5]{GJL2021})}
Let $D$ be a DAG, and define the pairs $\lambda(v) = (\floor{\log_2\pi(r,v)}, \floor{\log_2\pi(v, W))})$. The \emph{symmetric centroid decomposition (SC-decomposition)} of a DAG $D$ produces a set of edges between nodes with the same $\lambda$ pairs defined as $E_{scd}(D) = \{(u, i, v) \,|\, \lambda(u) = \lambda(v)\}$,
partitioning $D$ into disjoint paths called \emph{SC-paths} (some of them possibly of length 0).
\end{definition}\bigskip

The set $E_{scd}$ can be computed in $\bo(|D|)$ time. If $D$ is the DAG of an SLP $G$, then $|D|$ is $O{(|G|)}$. The following lemma justifies the name ``SC-paths''.

\bigskip\begin{lemma}{(Ganardi et al. \cite[Lemma~2.1]{GJL2021})}\label{lemma:scd}
Let $D = (V , E)$ be a DAG. Then every node has at most one outgoing and at most one incoming edge from $E_{scd}(D)$. Furthermore, every path from the root r to a sink node contains at most $2\log_2 n(D)$ edges that do not belong to $E_{scd}(D)$.
\end{lemma}\bigskip

Note that the sum of the lengths of all SC-paths is at most the number of nodes of the DAG, or equivalently, the number of variables of the SLP.

The following definition and technical lemma are needed to construct the building blocks of our balanced GSLPs.

\bigskip\begin{definition}{(Ganardi et al. \cite[page~7]{GJL2021})}
A \emph{weighted string} is a string $T \in \Sigma^*$ equipped with a \emph{weight function} $||\cdot||: \Sigma \rightarrow \mathbb{N} \backslash \{0\}$, which is extended homomorphically. If $A$ is a variable in an SLP $G$, then we write $||A||$ for the weight of the string $\gexp(A)$ derived from $A$.
\end{definition}\
\bigskip\begin{lemma}{(Ganardi et al. \cite[Proposition~2.2]{GJL2021})}\label{lemma:weighted}
For every non-empty weighted string $T$ of length $n$ one can construct in linear time an SLP $G$ generating $T$ with the following properties:
\begin{itemize}
    \item $G$ contains at most $3n$ variables
    \item All right-hand sides of $G$ have length at most 4
    \item $G$ contains suffix variables $S_1 , . . . , S_n$ producing all non-trivial suffixes of $T$
    \item every path from $S_i$ to some terminal symbol $a$ in the derivation tree of $G$ has length at most
$3 + 2(\log_2 ||S_i|| - \log_2 ||a||)$
\end{itemize}
\end{lemma}

With this machinery, we are ready to prove the main result of this section.

\bigskip\begin{theorem}\label{thm:balancing}Given a balanceable GSLP $G$ generating a string $T$, it is possible to construct an equivalent GSLP $G'$ of size $\bo(|G|)$ and height $\bo(\log n)$ in $O(|G| + t(G))$ time, where $t(G)$ is the time needed to compute the lengths of the expansion of each variable in $G$. \end{theorem}\bigskip

\begin{proof}Transform the GSLP $G$ into an SLP $H$ by replacing their special rules $A \rightarrow x$ by $A \rightarrow \mathtt{OUT}(x)$ (conceptually), and then obtain the SC-decomposition $E_{scd}(D)$ of the DAG $D$ of $H$. Observe that the SC-paths of $H$ use the same variables of $G$, so it holds that the sum of the lengths of all the SC-paths of $H$ is less than the number of variables of $G$. Also, note that any special variable $A \rightarrow x$ of $G$ is necessarily the endpoint (i.e., the last node of a directed path) of an SC-path in $D$. To see this note that $\lambda(A) \not= \lambda(B)$ for any $B$ that appears in $\mathtt{OUT}(x)$, because $\log_2 \pi(A,W) \geq \log_2 (|\mathtt{OUT}(x)|_B \cdot \pi(B,W)) \ge 1 + \log_2 \pi(B,W)$ where $ |\mathtt{OUT}(x)|_B \ge 2$ because $G$ is balanceable. This implies that the balancing procedure of Ganardi et al. on $H$, which transforms the rules of variables that are not the endpoint of an SC-path in the DAG $D$, will not touch variables that were originally special variables in $G$.

Let $\rho=(A_0 , d_0 , A_1 ), (A_1 , d_1 , A_2 ), \dots , (A_{p-1} , d_{p-1} , A_p)$ be an SC-path of $D$. It holds that for each $A_i$ with $i \in [0\dd p-1]$, in the SLP $H$ its rule goes to two distinct variables, one to the left and one to the right. Thus, for each variable $A_i$, with $i \in [0\dd p-1]$, there is a variable $A_{i+1}'$ that is not part of the path. Let $A_1'A_2'\dots A_p'$ be the sequence of these variables. Let $L = L_1L_2\dots L_s$ be the subsequence of left variables of the previous sequence. Then construct an SLP of size $\bo(s) \subseteq \bo(p)$ for the sequence $L$ (seen as a string) as in Lemma \ref{lemma:weighted}, using $|\gexp(L_i)|$ in $H$ as the weight function. In this SLP, any path from the suffix nonterminal $S_i$ to a variable $L_j$ has length at most $3 + 2(\log_2 ||S_i|| - \log_2 ||L_j||)$. Similarly, construct an SLP of size $\bo(t)\subseteq \bo(p)$ for the sequence $R = R_1R_2\dots R_t$ of right symbols in reverse order, as in Lemma \ref{lemma:weighted}, but with prefix variables $P_i$ instead of suffix variables. Each variable $A_i$, with $i \in [0\dd p-1]$, derives the same string as $w_{l}A_pw_{r}$, for some suffix $w_{l}$ of $L$ and some prefix $w_{r}$ of $R$. We can find rules deriving these prefixes and suffixes in the SLPs produced in the previous step, so for any variable $A_i$, we construct an equivalent rule of length at most 3. Add these equivalent rules, and the left and right SLP rules to a new GSLP $G'$. Do this for all SC-paths. Finally, add the original terminal variables and special variables (which are left unmodified) of the GSLP $G$, so $G'$ is a GSLP equivalent to $G$. 

Figure~\ref{fig:balancing} shows an example where the special GSLP rules are of the form $A \rightarrow B^t$, meaning $t$ copies of $B$ (i.e., the GSLP is an RLSLP).

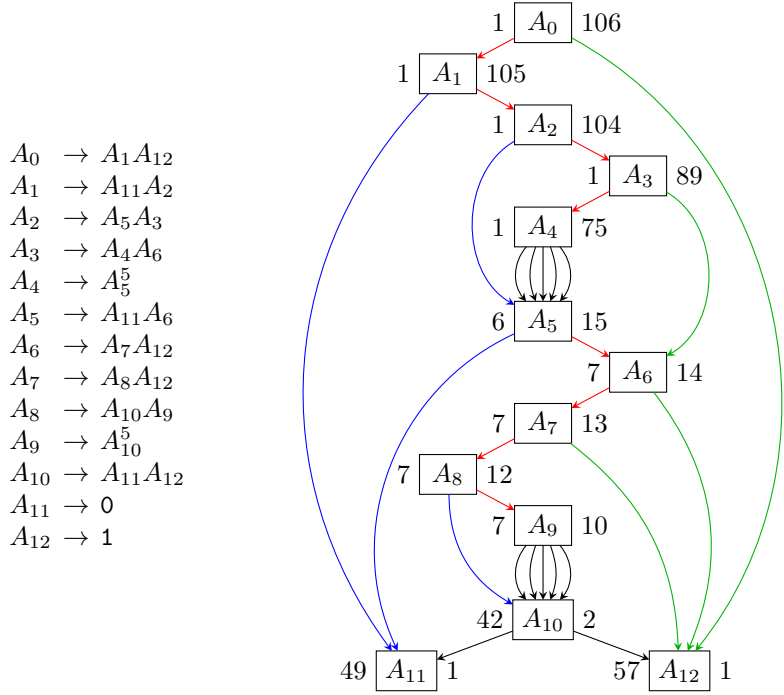
\begin{figure}[t]
    \centering
    \begin{minipage}{0.25\textwidth}
    $\begin{array}{lcl}
         A_0 & \rightarrow & A_1A_{12}  \\
         A_1 & \rightarrow & A_{11}A_2 \\
         A_2 & \rightarrow & A_5A_3  \\
         A_3 & \rightarrow & A_4A_6  \\
         A_4 & \rightarrow & A_5^5   \\
         A_5 & \rightarrow & A_{11}A_6  \\
         A_6 & \rightarrow & A_7A_{12} \\
         A_7 & \rightarrow & A_8A_{12} \\
         A_8 & \rightarrow & A_{10}A_9 \\
         A_9 & \rightarrow & A_{10}^5 \\
         A_{10} & \rightarrow & A_{11}A_{12} \\
         A_{11} & \rightarrow & \texttt{0} \\
         A_{12} & \rightarrow & \texttt{1}
         \end{array}$
    \end{minipage}
    \begin{minipage}{0.55\textwidth}
    \begin{tikzpicture}[->, >=stealth,node distance={1.25cm}, main/.style = {draw, rectangle, outer sep=0pt, minimum width=0.75cm, minimum height=0.4cm}]
    \node[main, label=left:1, label=right:106] (0) {$A_0$};
    \node[main, label=left:1, label=right:105] (1) [below left=0.15cm and 0.5cm of 0] {$A_1$};
    \node[main, label=left:1, label=right:104] (2) [below right=0.15cm and 0.5cm of 1] {$A_2$};
    \node[main, label=left:1, label=right:89] (3) [below right=0.15cm and 0.5cm of 2] {$A_3$};
    \node[main, label=left:1, label=right:75] (4) [below left=0.15cm and 0.5cm of 3] {$A_4$};
    \node[main, label=left:6, label=right:15] (5) [below of=4] {$A_5$};
    \node[main, label=left:7, label=right:14] (6) [below right=0.15cm and 0.5cm of 5] {$A_6$};
    \node[main, label=left:7, label=right:13] (7) [below left=0.15cm and 0.5cm of 6] {$A_7$};
    \node[main, label=left:7, label=right:12] (8) [below left=0.15cm and 0.5cm of 7] {$A_8$};
    \node[main, label=left:7, label=right:10] (9) [below right=0.15cm and 0.5cm of 8] {$A_9$};
    \node[main, label=left:42, label=right:2] (10) [below of=9] {$A_{10}$};
    \node[main, label=left:49, label=right:1] (11) [below left=0.15cm and 1cm of 10] {$A_{11}$};
    \node[main, label=left:57, label=right:1] (12) [below right=0.15cm and 1cm of 10] {$A_{12}$};
    % edges not in an SC-path nor connecting two disjoint paths
    \draw[red] (0) to (1);
    \draw[red] (1) to (2);
    \draw[red] (2) to (3);
    \draw[red] (3) to (4);
    \draw[bend right=20] (4) to (5);
    \draw[bend right=40] (4) to (5);
    \draw (4) to (5);
    \draw[bend left=20] (4) to (5);
    \draw[bend left=40] (4) to (5);
    \draw[red] (5) to (6);
    \draw[red] (6) to (7);
    \draw[red] (7) to (8);
    \draw[red] (8) to (9);
    \draw[bend right=20] (9) to (10);
    \draw[bend right=40] (9) to (10);
    \draw(9) to (10);
    \draw[bend left=20] (9) to (10);
    \draw[bend left=40] (9) to (10);
    \draw (10) to (11);
    \draw (10) to (12);
    % blue edges
    \draw[bend right=40, blue] (1) to (11);
    \draw[bend right=60, blue] (2) to (5);
    \draw[bend right=45, blue] (5) to (11);
    \draw[bend right=30, blue] (8) to (10);
    % green edges
    \draw[bend left=50, color={black!30!green}] (0) to (12);
    \draw[bend left=60, color={black!30!green}] (3) to (6);
    \draw[bend left=30, color={black!30!green}] (6) to (12);
    \draw[bend left=25, color={black!30!green}] (7) to (12);
    \end{tikzpicture}
    \end{minipage}
    \medskip
    \caption{The DAG and SC-decomposition of an unfolded RLSLP generating the string $\zero(\zero(\zero\one)^6\one^2)^6(\zero\one)^5\one^3$. The value to the left of a node is the number of paths from the root to that node, and the value to the right is the number of paths from the node to sink nodes. Red edges belong to the SC-decomposition of the DAG. Blue (resp. green) edges branch from an SC-path to the left (resp. to the right).}
    \label{fig:balancing}
\end{figure}

The SLP constructed for $L$ has all its rules of length at most 4, and $3s \leq 3p$ variables. The same happens with $R$. The other constructed rules also have a length of at most 3, and there are $p$ of them. Summing over all SC-paths, we have $\bo(|G|)$ size. The special variables cannot sum up to more than $\bo(|G|)$ size. Thus, the GSLP $G'$ has size $\bo(|G|)$.

Any path in the derivation tree of $G'$ is of length $\bo(\log n)$. To see why, let $A_0,\dots,A_p$ be an SC-path. Consider a path from a variable $A_i$ to an occurrence of a variable that is in the right-hand side of $A_p$ in $G'$. Clearly, this path has length at most 2. Now consider a path from $A_i$ to a variable $A_j'$ in $L$ with $i < j \leq p$. By construction this path is of the form $A_i\rightarrow S_k \rightarrow^* A_j'$ for some suffix variable $S_k$ (if the occurrence of $A_j'$ is a left symbol), and its length is at most $1 + 3 + 2(\log_2 ||S_k|| - \log_2 ||A_j'||) \leq4+2\log_2||A_i||-2\log_2||A_{j}'||$. Analogously, if $A_j'$ is a right variable, the length of the path is bounded by $1 + 3 + 2(\log_2 ||P_k|| - \log_2 ||A_j'||) \leq4+2\log_2||A_i||-2\log_2||A_{j}'||$. Finally, consider a maximal path to a leaf in the derivation tree of $G'$. Factorize it as
$$A_0 \rightarrow^* A_1 \rightarrow^* \dots \rightarrow^* A_{k}$$ 
where each $A_i$ is a variable of $H$ (and also of $G$). Paths $A_i \rightarrow^* A_{i+1}$ are like those defined in the paragraph above, satisfying that their length is bounded by $4+2\log_2||A_i||-2\log_2||A_{i+1}||$. Observe that between each $A_i$ and $A_{i+1}$, in the DAG $D$ there is almost an SC-path, except that the last edge is not in $E_{scd}$. The length of the path from root to leaf in $G'$ is at most
$$\sum_{i=0}^{k-1}(4+2\log_2||A_i||-2\log_2||A_{i+1}||) \leq 4k + 2\log_2||A_0|| - 2\log_2||A_k||$$

By Lemma \ref{lemma:scd}, $k \leq 2\log_2 n$, which yields the upper bound $\bo(\log n)$. 

To have standard SLP rules of size at most two, delete rules in $G'$ of the form $A \rightarrow B$ (replacing all $A$'s by $B$'s), and note that rules of the form $A \rightarrow BCDE$ or $A \rightarrow BCD$ can be decomposed into rules of length $2$, with only a constant increase in size and depth.

The balancing procedure uses $\bo(|G| + t(G))$ time and $\bo(|G| + s(G))$ auxiliary space, where $t(G)$ and $s(G)$ are the time and space needed to compute and store the set of all the pairs $(B, |\texttt{OUT}(x)|_B)$, where $B$ appears $|\texttt{OUT}(x)|_B>0$ times in $\texttt{OUT}(x)$, for every special variable $A \rightarrow x$. With this information the set $E_{scd}(D)$ can be computed in $\bo(|G|+t(G))$ time, instead of $\bo(|H|)$ time. The SLPs of Lemma \ref{lemma:weighted} are constructed in linear time in the lengths of the SC-paths, which sum to $\bo(|G|)$ in total.
\end{proof}

\section{Iterated Straight-Line Programs}

We now define iterated SLPs and show that they can be much smaller than $\delta$.

\bigskip\begin{definition}
An \emph{iterated straight-line program} of \emph{degree} $d$ ($d$-ISLP)  is an SLP that allows in addition \emph{iteration rules} of the form $$A \rightarrow \prod_{i=k_1}^{k_2} B_1^{i^{c_1}} \cdots B_t^{i^{c_t}}$$ 
where $1 \le k_1, k_2$, $0 \le c_1,\dots,c_t \leq d$ are integers and $B_1 \dots B_t$ are variables that cannot reach $A$ (so the ISLP generates a unique string). Iteration rules have size $2+2t=\bo(t)$ and expand to
$$\gexp(A) = \prod_{i=k_1}^{k_2} \gexp(B_1)^{i^{c_1}}\!\cdots \gexp(B_t)^{i^{c_t}}$$
where if $k_1 > k_2$ the iteration goes from $i=k_1$ downwards to $i=k_2$.
The size $size(G)$ of a $d$-ISLP $G$ is the sum of the sizes of all of its rules.
\end{definition}

\bigskip\begin{definition}
The measure $g_{it(d)}(T)$ is defined as the size of the smallest $d$-ISLP that generates $T$, whereas $g_{it}(T) = \min_{d \ge 0} g_{it(d)}(T)$.
\end{definition}\bigskip

The following observations show that ISLPs subsume RLSLPs, and thus, can be smaller than the smallest L-system.

\bigskip
\begin{proposition}\label{prop:git<grl}For any $d \geq 0$, it always holds that $g_{it(d)} = \bo(g_{rl})$.
\end{proposition}

\begin{proof}Just note that a rule $A \rightarrow \prod_{i=1}^t B^{i^0}$ from an ISLP simulates a rule $A \rightarrow B^t$ from a RLSLP. In particular, $0$-ISLPs are equivalent to RLSLPs. 
\end{proof}

\begin{proposition}For any $d \geq 0$, there exists a string family where $g_{it(d)} =o(\ell)$.
\end{proposition}

\begin{proof}Navarro and Urbina show a string family where $g_{rl} = o(\ell)$ \cite{NU23}. Hence, $g_{it(d)}$ is also $o(\ell)$ in this family. 
\end{proof}

We now show that $d=1$ suffices to obtain ISLPs that are significantly smaller than $\delta$ for some string families.

\bigskip\begin{lemma}Let $d \ge 1$. There exists a string family with $g_{it(d)} = \bo(1)$ and $\delta = \Omega(\sqrt{n})$.
\end{lemma}

\begin{proof}Such a family is formed by the strings $s_k = \prod_{i=1}^k \syma^i\symb$. The 1-ISLPs with initial rule $S_k \rightarrow \prod_{i=1}^{k} A^iB$, and rules $A \rightarrow \syma$, $B \rightarrow \symb$, generate each string $s_k$ in the family using $\bo(1)$ space. On the other hand, it has been proven that $\delta = \Omega(\sqrt{n})$ in the family $\symc s_k$ \cite{NU23}. As $\delta$ can only decrease by $1$ after the deletion of a character \cite{AFI2023},  $\delta = \Omega(\sqrt{n})$ in the family $s_k$ too.
\end{proof}

On the other hand, ISLPs can perform worse than other compressed representations; recall that $\delta \le \gamma \le b \le r_\dol$.

\bigskip\begin{lemma}\label{lem:fib} Let  $\mu \in \{r, r_\dol, \ell\}$. For any $d \ge 0$, there exists a string family with $g_{it(d)} = \Omega(\log n)$ and $\mu = \bo(1)$.
\end{lemma}

\begin{proof}
Consider the family of Fibonacci words defined recursively as $F_0 = \syma$, $F_1 = \symb$, and $F_{i+2} = F_{i+1}F_i$ for $i \geq 0$. Fibonacci words cannot contain substrings of the form $x^4$ for any $x \neq \varepsilon$ \cite{KARHUMAKI1983}. Consider an ISLP for a Fibonacci word and a rule of the form $A \rightarrow \prod_{i=k_1}^{k_2}B_1^{i^{c_1}}\cdots B_t^{i^{c_t}}$. Observe that if $c_r \neq 0$ for some $r$, then $\max(k_1, k_2) < 4$, as otherwise $\gexp(B_r)^4$ occurs in $T$. Similarly, if $c_r = 0$ for all $r$, then $|k_1 - k_2| < 3$, as otherwise $\gexp(B_1\cdots B_t)^4$ appears in $T$. In the latter case, we can rewrite the product with $k_1, k_2 \in [1\dd 3]$. Therefore, we can unfold the product rule into standard SLP rules of total size at most $9t$ ($3t$ variables raised to at most 3 each because we assumed our word is Fibonacci). Hence, for any $d$-ISLP $G$ generating a Fibonacci word, there is an SLP $G'$ of size $\bo(|G|)$ generating the same string. As $g = \Omega(\log n)$ in every string family \cite{Navacmcs20.3}, we obtain that $g_{it(d)} = \Omega(\log n)$ in this family too. On the other hand, $r_\dol, r$, and $\ell$ are $\bo(1)$ in the even Fibonacci words \cite{NOPtit20,MRS2003,NU21}.
\end{proof}

\bigskip\begin{lemma}For any $d \ge 0$, there exists a string family satisfying that $z = \bo(\log n)$ and $g_{it(d)} = \Omega(\log^2 n/\log\log n)$.
\end{lemma}

\begin{proof}Let $T(n)$ be the length $n$ prefix of the infinite Thue-Morse word on the alphabet $\{\syma,\symb\}$. Let $k_1,...,k_p$ be a set of distinct positive integers, and consider strings of the form $S = T(k_1)|_1T(k_2)|_2\cdots T(k_{p-1})|_{p-1}T(k_p)$, where $|_i$'s are unique separators and $k_1$ is the largest of the $k_i$. Since the sequences $T(k_i)$ are cube-free \cite{AlloucheShallit_ThueMorse}, there is no asymptotic difference in the size of the smallest SLP and the smallest ISLP (similarly to Lemma \ref{lem:fib}) for the string $S$. Hence, $g_{id(d)} = \Theta(g)$ in this family. It has been proven that $g = \Omega(\log^2 k_1 /\log\log k_1)$ and $z = \bo(\log k_1)$ for some specific sets of integers where $p = \Theta(k_1)$ \cite{BGLP2018}. Thus, the result follows.
\end{proof}

One thing that makes ISLPs robust is that they are not very sensitive to reversals, morphism application, or edit operations (insertions, deletions, and substitutions of a single character). The measure $g_{it(d)}$ behaves similarly to SLPs in this matter, for which it has been proved that $g(T') \le 2g(T)$ after an edit operation that converts $T$ to $T'$ \cite{AFI2023}, and that $g(\varphi(T)) \le g(T) + c_{\varphi}$ with $c_{\varphi}$ a constant depending only on the morphism $\varphi$ \cite{FRMU2023}. This makes $g_{it(d)}$ much more robust to string operations than measures like $r$ and $r_\$$, which are highly sensitive to all these transformations \cite{GILPST2021,GILRSU2023,FRMU2023,AFI2023}.

\bigskip\begin{lemma} \label{lem:edit}
Let $G$ be a $d$-ISLP generating $T$. Then there exists a $d$-ISLP of size $|G|$ generating the reversed text $T^R$. Let $\varphi$ be a morphism. Then there exists a $d$-ISLP of size $|G| + c_\varphi$ generating the text $\varphi(T)$, where $c_\varphi$ is a constant depending only on $\varphi$. Moreover, there exists a $d$-ISLP of size at most $\bo(|G|)$ generating $T'$ where $T$ and $T'$ differ by one edit operation.
\end{lemma}

\begin{proof}For the first claim,  note that reversing all the SLP rules and expressions inside the special rules, and swapping the values $k_1$ and $k_2$ in each special rule is enough to obtain a $d$-ISLP of the same size generating $T^R$. 

For the second claim, we replace rules of the form $A \rightarrow \syma$ with $A \rightarrow \varphi(\syma)$, yielding a grammar of size less than $|G| + \sum_{a \in \Sigma}|\varphi(a)|$. Then we replace these rules with binary rules, which asymptotically do not increase the size of the grammar.

For the edit operations, we proceed as follows. Consider the derivation tree of the ISLP, and the path from the root to the character we want to substitute, delete, or insert a character before or after. Then, we follow this path in a bottom up manner, constructing a new variable $A'$ for each node $A$ we visit. We start at some $A \rightarrow \syma$, so we construct $A' \rightarrow x$ where either $x = \symc$ or $x = \syma\symc$ or $x= \symc\syma$ or $x = \varepsilon$ depending on the edit operation. If we reach a node $A \rightarrow BC$ going up from $B$ (so we already constructed $B'$), we construct a node $A' \rightarrow B'C$ (analogously if we come from $C$).
If we reach a node $A \rightarrow \prod_{i=k_1}^{k_2} B_1^{i^{c_1}}\dots B_t^{i^{c_t}}$ going up from a specific $B_r$ with $r \in [1\dd t]$ (so we already constructed $B_r'$) at the $k$-th iteration of the product with $k_1 \le k \le k_2$ and being the $q$-th copy of $B_r$ inside $B_r^{k^{c_r}}$, then we construct the following new rules 
\begin{align*}
&A_1  \rightarrow \prod_{i=k_1}^{k-1}B_1^{i^{c_1}}\dots B_t^{i^{c_t}},\, A_2  \rightarrow \prod_{i=k}^{k}B_1^{i^{c_1}}\dots B_{r-1}^{i^{c_{r-1}}},\, A_3  \rightarrow \prod_{i=1}^{{q-1}}B_r^{i^0}, \\ 
&A_4  \rightarrow \prod_{i=q+1}^{k^{c_r}}B_r^{i^0},\,
A_5  \rightarrow \prod_{i=k}^{k}B_{r+1}^{i^{c_{r+1}}}\dots B_t^{i^{c_t}},\,
A_6  \rightarrow \prod_{i=k+1}^{k_2}B_1^{i^{c_1}}\dots B_t^{i^{c_t}}\\
&A' \rightarrow A_1A_2A_3B_r'A_4A_5A_6
\end{align*}
which are equivalent to $A$ (except by the modified, inserted, or deleted symbol) and sum to a total size of at most $6t + 21$. As $t \geq 1$, it holds that $(6t + 21)/(2t+2) \le 7$. After finishing the whole process, we obtain a $d$-ISLP of size at most $8|G|$. Note that this ISLP contains $\varepsilon$-rules. It also contains some non-binary SLP rules, which can be transformed into binary rules, at most doubling the size of the grammar.  
\end{proof}

\section{Accessing ISLPs}

We have shown that $g_{it(d)}$ breaks the lower bound $\delta$ already for $d \ge 1$. We now show that the measure is accessible. Concretely, we will prove the following result.

\bigskip\begin{theorem}
Let $T[1\dd n]$ be generated by a $d$-ISLP $G$ of height $h$. Then, we can build in time $\bo((|G|+d) d \lceil d\log d/\log n\rceil)$ and space $\bo(|G|+d\lceil d\log d/\log n\rceil)$ a data structure of size $\bo(|G|)$ that extracts any substring of $T$ of length $\lambda$ in time $\bo(\lambda + (h+\log n+d)d\lceil d\log d/\log n\rceil)$ on a RAM machine of $\Theta(\log n)$ bits, using $\bo(h+d\lceil d\log d/\log n\rceil)$ additional words of working space.
\end{theorem}\bigskip

Before proving this theorem, let us specialize its result by bounding $h$ and $d$. First, we show we can always make $h=\bo(\log n)$.

\bigskip\begin{lemma} \label{lem:balancing-islp}
Given a $d$-ISLP $G$ generating a string $T[1\dd n]$, it is possible to construct a $d'$-ISLP $G'$ of size $\bo(|G|)$ that generates $T$, for some $d' \le d$, with height $h' = \bo(\log n)$. The construction requires $\bo((|G|+d) d \lceil d\log d/\log n\rceil)$ time and $\bo(|G|+d\lceil d\log d/\log n\rceil)$ space.
\end{lemma}
\begin{proof}
ISLPs are GSLPs: they allow rules of the form $A \rightarrow \prod_{i=k_1}^{k_2} B_1^{i^{c_1}}\cdots B_t^{i^{c_t}}$ of size $2+2t$, and a simple program of size $\bo(t)$ writes the corresponding right-hand symbols (a sequence over $\{ B_1,\ldots,B_t\}$) explicitly. Note that, if $k_1 \neq k_2$ for every special rule, then the corresponding GSLP is balanceable for sure, as no symbol in any output sequence can appear exactly once. If $k_1=k_2$ for some special rule, instead, the output may have unique symbols $B_j^{i^0}$ or $B_j^{1^{c_j}}$. In this case we can split the rule at those symbols, in order to ensure that they do not appear in special rules, without altering the asymptotic size of the grammar. For example, $A\rightarrow B_1^{i^2} B_2^{i^0} B_3^{i^3} B_4^{i} B_5^{i^0}$ (i.e., $k_1=k_2=i$ for some $i>1$) can be converted into $A \rightarrow A_1 A_2$, $A_1 \rightarrow A_3 B_2$, $A_2 \rightarrow A_4 B_5$, $A_3 \rightarrow B_1^{i^2}$, $A_4 \rightarrow B_3^{i^3} B_4^i$. The case $k_1=k_2=1$ corresponds to $A \rightarrow B_1 \cdots B_t$ and can be decomposed into normal binary rules within the same asymptotic size.

We can then apply Theorem~\ref{thm:balancing}. Note the exponents of the special rules $A \rightarrow x$ are retained in general, though some can disappear in the case $k_1=k_2=1$. Thus, the parameter $d'$ of the balanced ISLP satisfies $d' \le d$.

The time to run the balancing algorithm is linear in $|G|$, except that we need to count, in the rules $A \rightarrow \prod_{i=k_1}^{k_2} B_1^{i^{c_1}}\cdots B_t^{i^{c_t}}$, how many occurrences of each nonterminal are produced. If we define 
\begin{equation} \label{eq:pc}
p_c(k)~=~\sum_{i=1}^k i^c,
\end{equation}
then $B_j$ is produced $p_{c_j}(k_2)-p_{c_j}(k_1-1)$ times on the right-hand side of $A$.

Computing $p_c(k)$ in the obvious way takes time $\Omega(k)$, which may lead to a balancing time proportional to the length $n$ of $T$. In order to obtain time proportional to the grammar size $|G|$, we need to process the rule for $A$ in time proportional to its size, $\bo(t)$. An alternative formula\footnote{See Wolfram Mathworld's {\tt https://mathworld.wolfram.com/BernoulliNumber.html}, Eqs.~(34) and (47).} computes $p_c(k)$ using rational arithmetic (note $c \le d$):
$$ p_c(k) ~~=~~ k^c + \frac{1}{c+1} \cdot \sum_{j=0}^c {c+1 \choose j}\, b_j \cdot k^{c+1-j}.$$
The formula requires $\bo(c) \subseteq \bo(d)$ arithmetic operations once the numbers $b_j$ are computed. Those $b_j$ are the Bernoulli (rational) numbers.  All the Bernoulli numbers from $b_0$ to $b_d$ can be computed in $\bo(d^2)$ arithmetic operations using the recurrence
$$ \sum_{j=0}^{d} {d+1 \choose j}\, b_j ~=~ 0,$$
from $b_0=1$. The numerators and denominators of the rationals $b_j$ fit in $\bo(j \log j) = \bo(d \log d)$ bits,\footnote{See {\tt https://www.bernoulli.org}, sections ``Structure of the denominator'', ``Structure of the nominator'', and ``Asymptotic formulas''.}
so they can be operated in time $\bo(\lceil d\log d/\log n\rceil)$ in a RAM machine with word size $\Theta(\log n)$. 

Therefore, once we build the Bernoulli rationals $b_j$ in advance, in time $\bo(d^2\lceil d\log d/\log n\rceil)$, the processing time for a rule of size $\bo(t)$ is $\bo(t\,d \lceil d\log d/\log n\rceil)$, which adds up to $\bo(|G|\,d \lceil d\log d/\log n\rceil)$ for all the grammar rules. Storing the precomputed values $b_j$ during construction requires $d \lceil d\log d/\log n\rceil$ extra space.
\end{proof}

The case $d=\bo(1)$ deserves to be stated explicitly because it yields near-optimal substring extraction time, and because it already breaks the space lower bound $\Omega(\delta)$. 

\bigskip\begin{corollary} \label{cor:extract1}
Let $T[1\dd n]$ be generated by a $d$-ISLP $G$, with $d=\bo(1)$. Then, we can build in $\bo(|G|)$ time and space a data structure of size $\bo(|G|)$ that extracts any substring of $T$ of length $\lambda$ in time $\bo(\lambda + \log n)$ on a RAM machine of $\Theta(\log n)$ bits. 
\end{corollary}\bigskip

Note that the corollary achieves $\bo(\log n)$ access time for a single symbol. Verbin and Yu \cite{VY13} showed that any data structure using space $s$ to represent $T[1\dd n]$ requires time $\Omega(\log^{1-\epsilon} n/\log s)$ time, for {\em any} $\epsilon>0$. Since even SLPs can use space $s=\bo(\log n)$ on some texts, they cannot always offer access time $\bo(\log^{1-\epsilon} n)$ for any constant $\epsilon$. This restriction applies to even smaller grammars like RLSLPs and $d$-ISLPs for any $d$.

For $d$ larger than $\bo(\log n)$, the next lemma  shows that we can always force $d$ to be $\bo(\log n)$ without asymptotically increasing the size of the grammar. From now on in the paper, we will disregard for simplicity the case $k_1 > k_2$ in the rules $A \rightarrow \Pi_{i=k_1}^{k_2} B_1^{i^{c_1}}\cdots B_t^{i^{c_t}}$, as their treatment is analogous to that of the case $k_1 \le k_2$.

\bigskip\begin{lemma} \label{lem:dlog}
If a $d$-ISLP $G$ generates $T[1\dd n]$, then there is also a $d'$-ISLP $G'$ of the same size that generates $T$, for some $d' \le \log_2 n$.
\end{lemma}
\begin{proof}
For any rule $A = \prod_{i=k_1}^{k_2} B_1^{i^{c_1}}\cdots B_t^{i^{c_t}}$, any $i \in [k_1\dd k_2]$, and any $c_j$, it holds that $n \ge |\gexp(A)| \ge i^{c_j}$, and therefore $c_j \le \log_i n$, which is bounded by $\log_2 n$ for $i \ge 2$. Therefore, if $k_2 \ge 2$, all the values $c_j$ can be bounded by some $d' \le \log_2 n$. A rule with $k_1=k_2=1$ is the same as $A \rightarrow B_1\cdots B_t$, so all values $c_j$ can be set to $0$ without changing the size of the rule.
\end{proof}

This yields our general result for arbitrary $d$.

\bigskip\begin{theorem}
Let $T[1\dd n]$ be generated by an ISLP $G$. Then, we can build in time $\bo((|G|+\log n)\log n\log\log n)$ and space $\bo(|G|+\log n)$ a data structure of size $\bo(|G|)$ that extracts any substring of $T$ of length $\lambda$ in time $\bo(\lambda + \log^2 n\log\log n)$ on a RAM machine of $\Theta(\log n)$ bits, using $\bo(\log n \log \log n)$ additional words of working space.
\end{theorem}\bigskip

\subsection{Data Structures}

We define some data structures that extend ISLPs $G$ allowing us to efficiently navigate them within $\bo(|G|)$ space. %Per Lemma~\ref{lem:dlog}, we assume $d = \bo(\log n)$.

Consider a rule $A \rightarrow \prod_{i=k_1}^{k_2} B_1^{i^{c_1}}\dots B_t^{i^{c_t}}$. The challenge is how to efficiently find the proper nonterminal $B_j^{i^{c_j}}$ on which we must recurse, while only using $\bo(t)$ space, which is proportional to the size of the rule. For example, we cannot afford a predecessor data structure on the $(k_2-k_1+1)t$ cumulative lengths of the runs the rule represents.

Though $t$ can be large, there are only $d+1$ distinct values $c_j$. We will exploit this because, per Lemma~\ref{lem:dlog}, $d$ can be made $\bo(\log n)$. 

\subsubsection{Navigating within a block}

To navigate within a ``block'' $B_1^{i^{c_1}} \cdots B_t^{i^{c_t}}$, for fixed $i$, we make use of auxiliary functions 
$$f_r(i) ~=~ \sum_{j=1}^{r}|\gexp(B_j)|\cdot i^{c_j},$$ 
for $r \in [1\dd t]$. In simple words, $f_r(i)$ computes the length of the expansion of $B_1^{i^{c_1}}\dots B_r^{i^{c_r}}$  (a prefix of the block), for a given $i$. 

We now show how to compute any $f_r(i)$ in time $\bo(d)$ using $\bo(t)$ space for each $A$.
We precompute an array $S_A[1\dd t]$ storing cumulative length information, as follows
$$ S_A[r] = \sum_{1 \le j \le r, c_j = c_r}  |\gexp(B_j)|.$$
That is, $S_A[r]$ adds up the lengths of the symbol expansions up to $B_r$ that must be multiplied by $i^{c_r}$. This array is easily built in time $\bo(t)$ once all the lengths $|\exp{\cdot}|$ have been computed, by traversing the rule left to right and maintaining in an array $L[B_j]$ the last position of each distinct nonterminal $B_j$ seen so far in the rule. Storing $L$ requires $\bo(t)$ space at construction time.

A second array, $C_A[1\dd t]$, stores the values $c_1,\ldots,c_t$. We preprocess $C_A$ to solve predecessor queries of the form
$$ \pred(A,r,c) = \max \{ j \le r,~C_A[j] = c\},$$
that is, the latest occurrence of $c$ in $C_A$ to the left of position $r$, for every $c=0,\ldots,d$. This query can be answered in $\bo(d)$ time because the elements in $C_A$ are also in $\{ 0,\ldots,d \}$: cut $C_A$ into chunks of length $d+1$, and for each chunk $C_A[(d+1)\cdot j +1 \dd (d+1)\cdot(j+1)]$ store precomputed values $\pred(A,(d+1)\cdot j,c)$ for all $c \in \{0,\ldots, d\}$. Those precomputed values are stored in $\bo(t)$ space and can be computed in $\bo(t)$ time, on a left-to-right traversal of the rule, using an array analogous to $L$, now indexed by the values $[0\dd d]$ (this requires $\bo(d)$ construction space).

To compute the values $r_c = \pred(A,r,c)$ for all $c$, find the chunk $j = \lceil r/(d+1) \rceil-1$ where $r$ belongs, initialize every $r_c = \pred(A,(d+1)\cdot j,c)$ for every $c$ (which is stored with the chunk $j$), and then scan the chunk prefix $C_A[(d+1)\cdot j+1 \dd r]$ left to right, correcting every $r_c \gets k$ if $c=C_A[k]$, for $k=(d+1)\cdot j+1 \dd r$.
Once the values $r_c$ are computed, we can  evaluate $f_r(i)$ in $\bo(d)$ time by adding up $S_A[r_c]\cdot i^c$ (because $S_A[r_c]$ adds up all $|\gexp(B_j)|$ that must be multiplied by $i^c$ in $f_r(i)$). Algorithm~\ref{alg:f_r} summarizes.

\begin{algorithm}[t]\caption{Computing $f_r(i)$ for nonterminal $A$, in time $\bo(d)$}\label{alg:f_r}
\begin{algorithmic}[1]
\Require Values $i$ and $r$, arrays $S_A$ and $C_A$, and precomputed values $\pred(A,(d+1) j,c)$ for every $j$ and $c$.
\Ensure The value $f_r(i)$.
%\Function{access}{$A, l$}
\State $j \gets \lceil r/(d+1) \rceil-1$
\For{$c \gets 0,\ldots,d$}
    \State $r_c \gets \pred(A,(d+1)j,c)$ (this is precomputed)
\EndFor
\For{$k \gets (d+1)j+1,\ldots,r$}
    \State $c \gets C_A[k]$ 
    \State $r_c \gets k$
\EndFor
\State $s \gets 0$
\State $p \gets 1$
\For{$c \gets 0,\ldots,d$}
    \State $s \gets s+ S_A[r_c]\cdot p$ 
    \State $p \gets p \cdot i$
\EndFor
\State \Return $s$
\end{algorithmic}
\end{algorithm}

\subsubsection{Navigating between blocks}

To select the proper ``block'', we define a second function: 
$$f^+(k) ~=~ \sum_{i = k_1}^{k}f_t(i),$$ 
so that $f^+(k)$ computes the cumulative sum of the length of the whole expressions $B_1^{i^{c_1}}\cdots B_t^{i^{c_t}}$ until $i=k$. The problem is that we cannot afford storing the function $f^+(k)$ in space proportional to the rule size.  
We will instead compute $f^+(k)$ by reusing the same data structures we store for $f_r(i)$: just as in Algorithm~\ref{alg:f_r}, for each $c=0,\ldots,d$, we compute $t_c = \pred(A,t,c)$ and $s_c = S_A[t_c]$. Instead of accumulating $s_c \cdot i^c$, however, we accumulate $s_c \cdot \sum_{i=k_1}^{k} i^c = s_c \cdot (p_c(k)-p_c(k_1-1))$, where $p_c(k)$ is defined in Eq.~(\ref{eq:pc}). 

\bigskip\begin{example}
Consider the ISLP of Proposition 2, defined by the rules
$S \rightarrow \prod_{i=1}^{k_2} A^iB$, $A \rightarrow \syma$, and $B \rightarrow \symb$. The polynomials associated with the representation of the rule $S$ are $i^{c_1} = i$ and $i^{c_2} = 1$. Then, we construct the auxiliary polynomials $f_1(i) = |\gexp(A)|i^{c_1} = i$ and $f_2(i) = |\gexp(A)| i^{c_1}+ |\gexp(B)| i^{c_2} = i + 1$. Finally, we construct the auxiliary polynomial $f^+(k) = \sum_{i=1}^k f_2(i) = \sum_{i=1}^k (i+1) = \frac{1}{2} k^2 + \frac{3}{2}k$. Figure~\ref{fig:polys} shows a more complex example to illustrate $C_A$ and $S_A$.
\end{example}
\bigskip

\begin{figure}[t]
\centering
\begin{tikzpicture}[array/.style={matrix of nodes,nodes={draw, minimum size=7mm, fill=gray!00},column sep=-\pgflinewidth, row sep=1mm, nodes in empty cells, row 1/.style={nodes={draw=none, fill=none, minimum size=7mm}}}]
\matrix[array] (array) {
1 & 2 & 3 & 4 & 5 & 6 & 7 & 8 & 9\\
2  & 3  & 6  & 7 & 14 & 13 & 5 & 3 & 18\\
1  & 2  & 1  &  0 & 0 & 1 & 2 & 3 & 0\\};
\draw (array-2-1.west)--++(0:0mm) node [left] (sa) {$S_A$};
\draw (array-3-1.west)--++(0:0mm) node [left] (ca) {$C_A$};
\draw (array-1-8.east)--++(0:0mm) node [right] (f8) {$\hspace{1cm}f_{r=8}(i) = 3i^3 + 5i^2 + 13i + 14$};
\draw (f8)--++(0:0mm) node  [below=0.5cm] (ft) {$\hspace{1cm}f_{t=9}(i) = 3i^3 + 5i^2 + 13i + 18$};
\draw (ft)--++(0:0mm) node [below=0.5cm] (fp) {$\hspace{2.1cm}f^+(k) = \frac{9}{12}k^4 + \frac{38}{12}k^3 + \frac{117}{12}k^2 + \frac{304}{12}k$};
\end{tikzpicture}
\caption{Data structures built for the ISLP rule $A \rightarrow \prod_{i=1}^5B^{i}C^{i^2}D^{i}EEE^iB^{i^2}C^{i^3}D$, with $|\gexp(B)| = 2$,  $|\gexp(C)| = 3$,  $|\gexp(D)| = 4$, and $|\gexp(E)| = 7$. We show some of the polynomials to be simulated with these data structures. }
\label{fig:polys}
\end{figure}

As shown right after Eq.~(\ref{eq:pc}), we can compute $p_c(k)$ (and hence $f^+(k)$) in time $\bo(d \lceil d\log d/\log n\rceil)$ after a preprocessing of time $\bo(d^2 \lceil d\log d/\log n\rceil)$.

\subsection{Direct Access}

We start with the simplest query: given the data structures of size $\bo(|G|)$ defined in the previous sections, return the symbol $T[l]$ given an index $l$.

For SLPs with derivation tree of height $h$, the problem is easily solved in $\bo(h)$ time by storing the expansion size of every nonterminal, and descending from the root to the corresponding leaf using $|\gexp(B)|$ to determine whether to descend to the left or to the right of every rule $A \rightarrow BC$. The general idea for $d$-ISLPs is similar, but now determining which child to follow in repetition rules is more complex. 

To access the $l$-th character of the expansion of $A \rightarrow \prod_{i=k_1}^{k_2} B_1^{i^{c_1}}\cdots B_t^{i^{c_t}}$ we first find the value $i$ such that $f^+(i-1) < l \le f^+(i)$ by using binary search (we let $f^+(i-1) = 0$ when $i = k_1$). Then, we find the value $r$ such that $f_{r-1}(i) < l -  f^+(i-1) \le f_r(i)$ by using binary search in the subindex of the functions (we let $f_{r-1}(i) = 0$ for any $i$ when $r = 1$). We then know that the search follows by $B_r$, with offset $l -  f^+(i-1) - f_{r-1}(i)$ inside $|\gexp(B_r)|^{i^{c_r}}$. The offset within $B_r$ is then easily computed with a modulus. 
Algorithm~\ref{alg:direct-access} gives the details, using $\suc$ to denote the binary search in an ordered set (i.e., $\suc([x_1\dd x_m],l)=j$ iff $x_{j-1} < l \le x_j$).

We carry out the first binary search so that, for every $i$ we try, if $f^+(i) < l$ we immediately answer $i+1$ if $l \le f^+(i+1)$; instead, if $l \le f^+(i)$, we immediately answer $i$ if $f^+(i-1) < l$. As a result, the search area is initially of length $|\gexp(A)|$ and, if the answer is $i$, the search has finished by the time the search area is of length $\le f^+(i)-f^+(i-1) = f_t(i)$. Thus, there are $\bo(1+\log(|\gexp(A)|/f_t(i)))$ binary search steps. The second binary search is modified analogously so that it carries out $\bo(1+\log(f_t(i)/(i^{c_r} |\gexp(B_r)|)))$ steps, for a total of at most $\bo(1+\log(|\gexp(A)|/|\gexp(B_r)|))$ steps. As the search continues by $B_r$, the sum of binary search steps telescopes to $\bo(h+\log n)$ on an ISLP of height $h$, and the total time is $\bo((h+\log n)\,d\lceil d\log d/\log n\rceil)$, plus the preprocessing time of $\bo(d^2\lceil d\log d/\log n\rceil)$. 

\begin{algorithm}[t]\caption{Direct access on $d$-ISLPs of height $h$ in $\bo((h+\log n+d)d)$ operations}\label{alg:direct-access}
\begin{algorithmic}[1]
\Require A variable $A$ of an ISLP, and a position $l \in [1, |\gexp(A)|]$.
\Ensure The character $\gexp(A)[l]$.
\Function{access}{$A, l$}
\If{$A \rightarrow a$}
    \State \Return $a$
\ElsIf{$A \rightarrow BC$}
    \If{$l \le |\gexp(B)|$}
        \State \Return \Call{access}{$B, l$}
    \Else { ($l > |\gexp(B)|$)}
        \State \Return \Call{access}{$C, l-|\gexp(B)|$}
    \EndIf
\Else{ ($A \rightarrow \prod_{i=k_1}^{k_2} B_1^{i^{c_1}} \dots B_t^{i^{c_t}}$)}
\State $i \leftarrow \suc([f^+(k_1)\dd f^+(k_2)], l)$
\State $l \leftarrow l - f^+(i-1)$
\State $r \leftarrow \suc([f_1(i)\dd f_t(i)], l)$
\State $l \leftarrow l - f_{r-1}(i)$
\State \Return \Call{access}{$B_r, (l-1 \bmod |\gexp(B_r)|)+1$}
\EndIf
\EndFunction
\end{algorithmic}
\end{algorithm}

\bigskip
\begin{example}We show how to access the $\symb$ at position 14 of the string $T = \prod_{i=1}^5\syma^i\symb$. Consider the ISLP $G$ and its auxiliary polynomials computed in Example 1. We start by computing $f^+(2) = 5$. As $l > 5$, we go right in the binary search and compute $f^+(4) = 14$. As $l \leq 14$ we go left, compute $f^+(3) = 9$ and find that $i = 4$. Hence, $T[l]$ lies in the expansion of $A^{i}B = A^4B$ at position $l_1 = l - f^+(i-1) = 5$. Then, we compute $f_1(4) = 4$. As $l_1 > 4$, we turn right and compute $f_2(4) = 5$, finding that $r = 2$. Hence, $T[l]$ lies in the expansion of $B^{i^0} = B^1$ at position $l_2 = l_1 - f_{r-1}(i) = 1$.
\end{example}

\subsection{Extracting Substrings}

Once we have accessed $T[l]$, it is possible to output the substring $T[l\dd l+\lambda-1]$ in $\bo(\lambda+h)$ additional time, as we return from the recursion in Algorithm~\ref{alg:direct-access}. We carry the parameter $\lambda$ of the number of symbols (yet) to output, which is first decremented when we finally arrive at line 3 and find the first symbol, $T[l]$, which we now output immediately. From that point, as we return from the recursion, instead of returning the symbol $T[l]$, we output up to $\lambda$ following symbols and return the number of remaining symbols yet to output, until $\lambda=0$. See Algorithm~\ref{alg:extract}.

\begin{algorithm}[t]\caption{Length-$\lambda$ substring access on ISLPs of height $h$ in $\bo(h+\lambda)$ extra time}\label{alg:extract}
\begin{algorithmic}[1]
\Require A variable $A$ of an ISLP, a position $l \in [1, |\gexp(A)|]$ and a length $\lambda>0$.
\Ensure Outputs $\exp{A}[l\dd l+\lambda-1]$ and returns the number of symbols it could not extract (if $l+\lambda-1 > |\exp{A}|$).
\Function{extract}{$A, l, \lambda$}
\If{$A \rightarrow a$}
    \State {\bf output} $a$
    \State $\lambda \gets \lambda-1$
\ElsIf{$A \rightarrow BC$}
    \If{$l \le |\gexp(B)|$}
        \State $\lambda \gets$ \Call{extract}{$(B,l,\lambda)$}
        \If{$\lambda>0$}
           \State $\lambda \gets$ \Call{extract}{$C, 1, \lambda)$}
        \EndIf
    \Else { ($l > |\gexp(B)|)$}
        \State $\lambda \gets$ \Call{extract}{$C, l-|\gexp(B)|,\lambda$}
    \EndIf
\Else{ ($A \rightarrow \prod_{i=k_1}^{k_2} B_1^{i^{c_1}} \dots B_t^{i^{c_t}}$)}
\State $i \leftarrow \suc([f^+(k_1)\dd f^+(k_2)], l)$
\State $l \leftarrow l - f^+(i-1)$
\State $r \leftarrow \suc([f_1(i)\dd f_t(i)], l)$
\State $l \leftarrow l - f_{r-1}(i)$
\State $\lambda \gets$ \Call{extract}{$B_r, (l-1 \bmod |\gexp(B_r)|)+1,\lambda$}
\State $k \gets \lceil l/|\exp{B_r}|\rceil+1$
\While {$i \le k_2 \land \lambda>0$}
   \While{$r \le t \land \lambda >0$}
      \While{$k \le i^{c_r} \land \lambda>0$}
         \State $\lambda \gets$ \Call{extract}{$B_r,1,\lambda$}
         \State $k \gets k+1$
       \EndWhile
       \State $k \gets 1$
       \State $r \gets r+1$
   \EndWhile
   \State $r \gets 1$
   \State $i \gets i+1$
\EndWhile
\EndIf
\State \Return $\lambda$
\EndFunction
\end{algorithmic}
\end{algorithm}

To analyze Algorithm~\ref{alg:extract}, we distinguish two types of recursive calls. Initially the substring to extract is within one of the children of the grammar tree node, and thus only one recursive call is made. Those are the cases of lines 7, 11, and 17. The number of those calls is limited by the the height $h$ of the grammar. Once we reach a node where the substring to extract spreads across more than one child, the $\lambda$ symbols to output are distributed across more than one recursive call, ending in line 3 when outputting individual symbols. Those recursive calls form a tree with no unary paths and $\lambda$ leaves, thus they add up to $\bo(\lambda)$. 

The total space for the procedure is $\bo(h)$ for the recursion stack (which is unnecessary when returning a single symbol, since recursion can be eliminated in that case), plus $\bo(d\lceil d\log d/\log n\rceil)$ for the precomputed Bernoulli rationals.\footnote{As these do not depend on the query, they could be precomputed at indexing time and be made part of the index, at a very modest increase in space.}

\subsection{Composable Functions on Substrings}

Other than extracting a text substring, we aim at computing more general functions on arbitrary ranges $T[p\dd q]$, in time that is independent of the length $q-p+1$ of the range. We show how to compute some functions that have been studied in the literature, focusing on {\em composable} ones.

\bigskip\begin{definition}
A function $f$ from strings is {\em composable} if there exists a function $h$ such that, for every pair of strings $X$ and $Y$, it holds $f(X\cdot Y)=h(f(X),f(Y))$.
\end{definition}\bigskip

We focus for now on two popular composable functions, which find applications for example on grammar-compressed suffix trees \cite{FMNtcs09,GNP18}.

\bigskip\begin{definition}
A {\em range minimum query (RMQ)} on $T[p\dd q]$ returns the leftmost position where the minimum value occurs in $T[p\dd q]$. Formally,
$$ \RMQ(T,p,q) ~=~ \min \{ k \in [p\dd q] \,|\, \forall k' \in [p\dd q], T[k] \le T[k']\}$$
\end{definition}

\begin{definition}
A {\em next/previous smaller value query (NSV/PSV)} on $T[p\dd n]$/$T[1\dd p]$ and with value $v$ find the smallest/largest position following/preceding $p$ with value at most $v$. Formally,
\begin{eqnarray*}
%\item $\psv(i)= \max(\{j \,|\, j < i, w[j] < w[i]\}\cup \{0\})$
%\item $\nsv(i) = \min(\{j \,|\, j > i, w[j] < w[i]\}\cup \{n+1\})$
\NSV(T,p, v) &=& \min(\{q \,|\, q \ge p, T[q] < v\}\cup \{n+1\}) \\
\PSV(T,p, v) &=& \max(\{q \,|\, q \le p, T[q] < v\}\cup \{0\})
\end{eqnarray*}
\end{definition}

\begin{algorithm}[t!]\caption{Range minimum queries on SLPs of height $h$ in $\bo(h)$ time}\label{alg:rmq}
\begin{algorithmic}[1]
\Require A variable $A$ of an SLP and positions $1 \le p \le q \le |\gexp(A)|$.
\Ensure Returns $\rmq(\exp{A}[p\dd q]$ and the corresponding minimum value.
\Function{rmq}{$A, p, q$}
\If{$(p,q) = (1,|\gexp(A)|)$}
    \Return $\rmq(A)$ (which is precomputed)
\ElsIf{$A \rightarrow BC$}
    \If{$q \le |\gexp(B)|$}
        \Return \Call{rmq}{$B,p,q$}
    \ElsIf{$p > |\gexp(B)|$}
        { \Return \Call{rmq}{$C,p-|\gexp(B)|,q-|\gexp(B)|$}}
    \Else { ($p \le |\gexp(B)| < q$)}
        \State $\langle m_l,v_l \rangle \gets$ \Call{rmq}{$B, p, |\gexp(B)|$}
        \State $\langle m_r,v_r \rangle \gets$ \Call{rmq}{$C,1,q-|\gexp(B)|$}
        \If{$v_l \le v_r$}
            \Return $\langle m_l,v_l\rangle$
        \Else
            { \Return $\langle |\gexp(B)|+m_r,v_r\rangle$}
        \EndIf
    \EndIf
\EndIf
\EndFunction
\end{algorithmic}
\end{algorithm}

We show next how to efficiently solve those queries on ISLPs. We do not know how to compute other more complex functions, like Karp-Rabin fingerprints \cite{kr:ibm87}, on ISLPs. This will be addressed in Section~\ref{sec:rlslps}, on the simpler RLSLPs.

\subsubsection{Range Minimum Queries}

Solving RMQs on an SLP $G$ is simple thanks to composability. More precisely, what is composable is an extended function $f(X) = \langle m,v,\ell\rangle$ where $m = \RMQ(X,1,|X|)$, $v=X[m]$, and $\ell=|X|$. Then, given $f(X)=\langle m_x,v_x,\ell_x\rangle$ and $f(Y)=\langle m_y,v_y,\ell_y\rangle$, it holds $f(X\cdot Y)=\langle m_x,v_x,\ell_x+\ell_y\rangle$ if $v_x \le v_y$, and $\langle \ell_x+m_y,v_y,\ell_x+\ell_y\rangle$ otherwise, which is computable in time $\bo(1)$.
Since $f(a)=(1,a,1)$ we also have $\bo(1)$ time.

To compute RMQs on an SLP $G$, we first preprocess the grammar to store $f(\exp{A})=\langle m,v,\ell\rangle$ for each nonterminal $A$, in the form of the pair $\rmq(A)=\langle m,v\rangle$ and the length $\ell=|\exp{A}|$. Thanks to the composability of $f$, this is is easily built in $\bo(|G|)$ time in a bottom-up traversal of the grammar. 

To solve $\RMQ(T[p\dd q])$ on the SLP, we descend from the root towards $T[p\dd q]$ (guided by the stored expansion lengths $|\exp{A}|$) until finding a leaf (if $p=q$), or more typically a rule $A \rightarrow BC$ such that $T[p\dd q] = \exp{B}[p'\dd |\exp{B}|] \cdot \exp{C}[1\dd q']$. At this point we split into two recursive calls, one computing $\RMQ$ on a suffix of $\exp{B}$ and another on a prefix of $\exp{C}$. By making the recursive calls return $\rmq(B)$ in $\bo(1)$ time when the range spans the whole string $\exp{B}$, we ensure that those prefix/suffix calls perform only one further (nontrivial) recursive call, and thus the query is solved in $\bo(h)$ time, traversing at most two root-to-leaf paths in the parse tree. Algorithm~\ref{alg:rmq} shows the details.

To solve RMQs on ISLPs, we observe that the expansion of $A \rightarrow \prod_{i=k_1}^{k_2} B_1^{i^{c_1}}\dots B_t^{i^{c_t}}$ always contains the same symbols. Further, the RMQ of $\exp{A}$ occurs always in the first block, $i=k_1$, and it depends essentially on the sequence $B_1 \cdots B_t$. To handle these rules, we preprocess them as follows. Let $\rmq(B_j)=\langle m_j,v_j\rangle$. Then, we build the string $v_1\cdots v_t$ and precompute an RMQ data structure on it that answers queries $\rmq_A(p,q) = \RMQ(v_1\cdots v_t,p,q)$. It is possible to build such a data structure in $\bo(t)$ time and bits of space, such that it answers queries in $\bo(1)$ time \cite{FH11} (so this adds just $\bo(|G|)$ time and bits to the grammar preprocessing cost). With this structure, we can simlulate the extension of our $\rmq(A)$ precomputed pairs to any subsequence of any $B_1^{i^{c_1}}\cdots B_t^{i^{c_t}}$: $\rmq(B_1^{i^{c_1}}\cdots B_t^{i^{c_t}},a,b) = \langle m,v\rangle$, where $\RMQ(v_1\cdots v_t,a,b)=m'$, $\rmq(B_{m'})=\langle m'',v\rangle$, and $m = f_{m'-1}(i)+m''$. The time this takes is dominated by the $\bo(d)$ cost to compute $f_{m'-1}(i)$.

At query time, when we arrive at such a node $A$ with limits $p$ and $q$, we proceed as in lines 10--13 of Algorithm~\ref{alg:direct-access} to find the values $i_p$ and $r_p$, and $i_q$ and $r_q$, corresponding to $p$ and $q$, respectively (just as we find $i$ and $r$ for $l$ as in Algorithm~\ref{alg:direct-access}). There are several possibilities:
\begin{enumerate}
\item If $i_p=i_q$ and $r_p=r_q$, then $p$ and $q$ fall inside $\exp{B_{r_p}^{i^{c_{r_p}}}}$. They may be both inside a single copy of $\exp{B_{r_p}}$, in which case we continue with a single recursive call. Or they may span a (possibly empty) suffix of $\exp{B_{r_p}}$, zero or more copies of $\exp{B_{r_p}}$, and a (possibly empty) prefix of $\exp{B_{r_p}}$. The query is then solved with at most two recursive calls on $B_{r_p}$ (which are prefix/suffix calls, and the information on $\rmq(B_{r_p})$. We compose as explained those (up to) three results, and add $f^+(i_p-1)+f_{r_p-1}(i_p)$ to the resulting position so as to place it within $\exp{A}$.
\item If $i_p=i_q$ and $r_p < r_q$, then we must also consider the subsequence $B_{r_p+1}^{i_p^{c_{r_p+1}}}\cdots B_{r_q-1}^{i_p^{c_{r_q-1}}}$, in case $r_q-r_p>1$. This additional candidate to the RMQ is found with $\rmq(B_1^{i_p^{c_1}}\cdots B_t^{i_p^{c_t}},r_p+1,r_q-1)$, in time $\bo(d)$ as explained.
\item If $i_p < i_q$, we must also add a suffix of  of $B_1^{i_p^{c_1}}\cdots B_t^{i_p^{c_t}}$, the whole $B_1^{(i_p+1)^{c_1}}\cdots B_t^{(i_p+1)^{c_t}}$ (if $i_q-i_p>1$), and a prefix of $B_1^{i_q^{c_1}}\cdots B_t^{i_q^{c_t}}$. All those are included with our simulation of queries $\rmq(B_1^{i^{c_1}}\cdots B_t^{i^{c_t}},a,b)$.
\end{enumerate}

Overall, we perform either one recursive call (when $p$ and $q$ are inside the same $B_{r_p}$, or two prefix/suffix recursive calls (for a suffix of $B_{r_p}$ and a prefix of $B_{r_q}$). The analysis is then the same as for the SLPs, except that we spend $\bo((h+\log n)d\lceil d\log d / \log n\rceil)$ time, plus the preprocessing time of $\bo(d^2\lceil d\log d/\log n\rceil)$, due to the binary searches needed to find $i_p$, $i_q$, $r_p$, and $r_q$, as for direct access.

\bigskip\begin{theorem}\label{thm:RMQ}
Let $T[1\dd n]$ be generated by a $d$-ISLP $G$ of height $h$. Then, we can build in time $\bo((|G|+d)d\lceil d\log d/\log n\rceil)$ and space $\bo(|G|+d\lceil d\log d/\log n\rceil)$ a data structure of size $\bo(|G|)$ that computes any query $\RMQ(T,p,q)$ in time $\bo((h+\log n+d)d \lceil d\log d/\log n\rceil)$ on a RAM machine of $\Theta(\log n)$ bits, using $\bo(h+d\lceil d\log d/\log n\rceil)$ additional words of working space.
\end{theorem}\bigskip

Since we can make both $h$ and $d$ be $\bo(\log n)$, we have the following corollary.

\bigskip\begin{corollary}\label{cor:RMQ}
Let $T[1\dd n]$ be generated by an ISLP $G$. Then, we can build in time $\bo((|G|+\log n)\log n\log\log n)$ and space $\bo(|G|+\log n\log\log n)$ a data structure of size $\bo(|G|)$ that computes any query $\RMQ(T,p,q)$ in time $\bo(\log^2 n\log\log n)$ on a RAM machine of $\Theta(\log n)$ bits, using $\bo(\log n\log\log n)$ additional words of working space.
\end{corollary}\bigskip

Finally, the following specialization is relevant, as for example it encompasses 1-ISLPs (which may break $\delta$) and RLSLPs.

\bigskip\begin{corollary}\label{cor:RMQ1}
Let $T[1\dd n]$ be generated by a $d$-ISLP $G$ with $d=\bo(1)$. Then, we can build in time and space $\bo(|G|)$ a data structure of size $\bo(|G|)$ that computes any query $\RMQ(T,p,q)$ in time $\bo(\log n)$ on a RAM machine of $\Theta(\log n)$ bits. %, using $\bo(\log n)$ additional words of working space.
\end{corollary}

\subsubsection{Next/Previous Smaller Value}

Let us consider query NSV; query PSV is analogous.
NSV is composable if we extend it to function $f(X,v)=\langle p,\ell\rangle$, where $p=\NSV(X,1,v)$ and $\ell=|X|$. If $f(X,v)=\langle p_x,\ell_x\rangle$ and $f(Y,v)=\langle p_y,\ell_y\rangle$, then $f(X\cdot Y)=\langle p,\ell_x+\ell_y\rangle$, where $p=p_x$ if $p_x \le \ell_x$, else $p=\ell_x+p_y$ if $p_y \le \ell_y$, and $p=\ell_x+\ell_y+1$ otherwise. The composition takes $\bo(1)$ time. %\marginpar{\textcolor{red}{Creo que $p$ no está bien definido aquí si $p_x \le \ell_x$ y $p_y \le \ell_y$ }}

The procedure to compute $\NSV(T,p,v)$ on an SLP is depicted in Algorithm~\ref{alg:nsv}. We reuse the precomputed pairs $\rmq(A)=\langle m,v\rangle$ of RMQs, using $\rmq(A).v$ to refer to $v$. Importantly, the algorithm uses that field to notice in constant time that the answer is not within $\exp{A}$ (lines 2--3). As for RMQs, the algorithm may perform two calls on $A \rightarrow BC$ when the call on $B$ fails, but then the call on $C$ is for $p=1$. A call with $p=1$ cannot perform two nontrivial recursive calls, because if the left symbol fails, this is noticed in constant time. So the recursion follows two paths at most: one of failing ``left'' symbols ($B$ in $A\rightarrow BC$) and one of ``right'' symbols that do not fail ($C$ in $A \rightarrow BC$). The total time is then $\bo(h)$. 

\begin{algorithm}[t!]\caption{Next smaller values on SLPs of height $h$ in $\bo(h)$ time}\label{alg:nsv}
\begin{algorithmic}[1]
\Require A variable $A$ of an SLP, position $1 \le p \le |\exp{A}|$, and threshold $v$.
\Ensure The position $\PSV(\exp{A},p,v)$.
\Function{nsv}{$A, p, v$}
\If{$\rmq(A).v \geq v$}
\State \Return $|\exp{A}|+1$
\ElsIf{$A \rightarrow a$}
\State{\Return $1$}
\ElsIf{$A \rightarrow BC$}
\If{$p \leq |\exp{B}|$}
\State{$p \gets$ \Call{nsv}{$B, p, v$}}
\If{$p \le |\exp{B}|$}
    \State \Return $p$
\EndIf
\EndIf
\State \Return $|\exp{B}|~ +$ \Call{nsv}{$C, p - |\exp{B}|, v$}
\EndIf
\EndFunction
\end{algorithmic}
\end{algorithm}

To extend the algorithm to ISLPs we must consider, as for the case of RMQs, the special rules. Just as in that case, the answer to a query $\NSV(\gexp(A),p,v)$ with $A \rightarrow \Pi_{i=k_1}^{k_2} B_1^{i^{c_1}}\cdots B_t^{i^{c_t}}$ depends essentially on the smallest values of the nonterminal expansions, $\gexp(B_j)$. Let again $\rmq(B_j)=\langle m_j,v_j\rangle$. We preprocess the string $v_1\cdots v_t$ to solve queries $\NSV(v_1\cdots v_j,p)$. This preprocessing takes $\bo(t\log t)$ time and $\bo(t)$ space, and answers NSV queries in time $\bo(\log^\epsilon t)$ for any constant $\epsilon>0$ \cite{NN12} (those are modeled as orthogonal range successor queries on a grid). 

We can then simulate precomputed values $\nsv(B_1^{i^{c_1}}\cdots B_t^{i^{c_t}},p,v)=q$, where $\NSV(v_1\cdots v_t,p,v)=q'$ is precomputed as explained, $\NSV(B_{q'},1,v)=q''$ is obtained with a recursive call, and $q=f_{q'-1}(i)+q''$. Note that the recursive call is for a whole symbol, and we are sure to find the answer inside it: if $\NSV(v_1\cdots v_t,p,v)=t+1$, we return $f_t(i)+1$ without making any recursive call.
At query time, after finding $i_p$ and $r_p$, we have the following cases:
\begin{enumerate}
\item We may have to recurse on a nonempty suffix of $B_{r_p}$, finishing if we find the answer inside it. If not, there may be more copies of $B_{r_p}$ ahead of position $p$, in which case we either determine in constant time that there is no answer inside $B_{r_p}$, or we recurse on the whole $B_{r_p}$ and find the answer inside it, thereby finishing.
\item If not finished, we may have to consider a block suffix $B_{r_p+1}^{i_p^{c_{r_p+1}}}\cdots B_t^{i_p^{c_t}}$. This is handled by computing $\nsv(B_1^{i_p^{c_1}}\cdots B_t^{i_p^{c_t}},r_p+1,v)$ as explained, possibly making a recursive call on a whole symbol, in which case we find the answer.
\item If not, we may find the answer in the next block, $B_1^{(i_p+1)^{c_1}}\cdots B_t^{(i_p+1)^{c_t}}$, in the same way as in point 2. If we find no answer here, then there is no answer to NSV.
\end{enumerate}

Just as for the case of SLPs, we traverse only two paths along the process. The main difference with the cost of RMQs is the $\bo(\log^\epsilon t) \subseteq \bo(\log^\epsilon |G|)$ cost incurred to compute orthogonal range successors.

\bigskip\begin{theorem}\label{thm:NSV}
Let $T[1\dd n]$ be generated by a $d$-ISLP $G$ of height $h$. Then, for any constant $\epsilon>0$, we can build in time $\bo(|G|(\log |G|+d\lceil d\log d/\log n\rceil) + d^2\lceil d\log d/\log n\rceil)$ and space $\bo(|G|+d\lceil d\log d/\log n\rceil)$ a data structure of size $\bo(|G|)$ that computes any query $\PSV/\NSV(T,p,v)$ in time $\bo(h\log^\epsilon|G|+(h+\log n+d)d \lceil d\log d/\log n\rceil)$ on a RAM machine of $\Theta(\log n)$ bits, using $\bo(h+d\lceil d\log d/\log n\rceil)$ additional words of working space.
\end{theorem}

\bigskip\begin{corollary}\label{cor:NSV}
Let $T[1\dd n]$ be generated by an ISLP $G$. Then, we can build in time $\bo((|G|+\log n)\log n\log\log n)$ and space $\bo(|G|+\log n\log\log n)$ a data structure of size $\bo(|G|)$ that computes any query $\PSV/\NSV(T,p,v)$ in time $\bo(\log^2 n\log\log n)$ on a RAM machine of $\Theta(\log n)$ bits, using $\bo(\log n\log\log n)$ additional words of working space.
\end{corollary}

\bigskip\begin{corollary}\label{cor:NSV1}
Let $T[1\dd n]$ be generated by a $d$-ISLP $G$ with $d=\bo(1)$. Then, for any constant $\epsilon>0$, we can build in time $\bo(|G|\log |G|)$ and space $\bo(|G|)$ a data structure of size $\bo(|G|)$ that computes any query $\PSV/NSV(T,p,v)$ in time $\bo(\log n \log^\epsilon |G|)$ on a RAM machine of $\Theta(\log n)$ bits. %, using $\bo(\log n)$ additional words of working space.
\end{corollary}

\section{Revisiting RLSLPs} \label{sec:rlslps}

As pointed out in Proposition \ref{prop:git<grl}, RSLPs are equivalent to $0$-ISLPs, because an ISLP rule $A \rightarrow \prod_{i=k_1}^{k_2} B^{i^0}$ corresponds exactly to the RLSLP rule $A \rightarrow B^{|k_2-k_1|+1}$. We can then apply Lemma~\ref{lem:balancing-islp} over any RLSLP to obtain an equivalent RLSLP of the same asymptotic size and height $\bo(\log n)$. Once we count with a balanced version of any RLSLP, we can reuse Corollaries~\ref{cor:extract1}, \ref{cor:RMQ1}, and \ref{cor:NSV1}, to obtain a similar result for RLSLPs. Note that we can improve those results because we do not need to preprocess the grammar to simulate the $\rmq$ and $\nsv$ queries on blocks, because in an RLSLP all the cases of run-length rules $A \rightarrow B^t$ fall inside the subcase 1 of RMQs and NSVs.

\bigskip\begin{corollary}\label{cor:RLSLP}
Let $T[1\dd n]$ be generated by a RLSLP $G$. Then, we can build in time and space $\bo(|G|)$ data structures of size $\bo(|G|)$ that (i) extract any substring $T[l\dd l+\lambda-1]$ in time $\bo(\lambda+\log n)$, (ii) compute any query $\RMQ(T,p,q)$ in time $\bo(\log n)$, and (iii) compute any query $\PSV/\NSV(T,p,v)$ in time $\bo(\log n)$, on a RAM machine of $\Theta(\log n)$ bits. %, using $\bo(\log n)$ additional words of working space.
\end{corollary}\bigskip

Those results on RLSLPs have already been obtained before \cite{GNP18,CEKNP20}, but our solutions exploiting balancedness are much simpler once projected into the run-length rules. We now exploit the simplicity of RLSLPs to answer a wider range of queries on substrings, and show as a particular case how to compute Karp-Rabin fingerprints in logarithmic time; we do not know how to do that on general ISLPs.

\subsection{More General Functions}\label{sec:general-functions}

We now expand our results to a wide family of composable functions that can be computed in $\bo(\log n)$ time on top of balanced RLSLPs. We prove the following result.

\begin{algorithm}[tp!]\caption{Computation of general string functions in RLSLPs in $\bo(\log n)$ steps}\label{alg:general}
\begin{algorithmic}[1]
\Require A variable $A$ of an RLSLP (with its arrays $L$ and $F$ as global variables), and two positions $1 \leq i \leq j \leq |\exp{A}|$.
\Ensure $f(\exp{A}[i \dd j])$.
\Function{f}{$A, i, j$}
\If{$(i,j) = (1,|\exp{A}|)$}
\State \Return $F[A]$
\ElsIf{$A \rightarrow BC$}
\If{$j \leq |\exp{B}|$}
\State{\Return \Call{f}{$B, i, j$}}
\ElsIf{$|\exp{B}| < i$}
\State {\Return \Call{f}{$C, i - |\exp{B}|, j - |\exp{B}|$}}
\Else{ ($i \leq |\exp{B}| < j$)}
\State{$f_l \leftarrow $ \Call{f}{$B, i, |\exp{B}|$}}
\State{$f_r \leftarrow $ \Call{f}{$C, 1, j - |\exp{B}|$}}
\State {\Return $h(f_l,f_r)$}
\EndIf
\Else{ ($A \rightarrow B^t$)}
\State $t' \gets \lceil i/|\exp{B}|\rceil$
\State $t'' \gets \lceil j/|\exp{B}|\rceil$
\If{$t' = t''$}
\State { \Return \Call{f}{$B, i - (t'-1)\cdot |\exp{B}|, j - (t'-1)\cdot |\exp{B}|$}}
\EndIf
\State {$f_l \leftarrow$ \Call{f}{$B, i - (t'-1)\cdot |\exp{B}|, |\exp{B}|$}}
\State {$f_r \leftarrow$ \Call{f}{$B, 1, j - (t''-1)\cdot |\exp{B}|$}}
\State{Compute $f_c(t''-t'-1)$ using the recurrence
$$
f_c(k) \leftarrow
\begin{cases}
f(\varepsilon) & \text{if }k = 0; \\
F[B] & \text{if }k = 1; \\
h(f_c(k/2),f_c(k/2)) &\text{if }k \text{ is even}; \hspace{3.5cm} \\
h(F[B], f_c(k-1)) &\text{if }k \text{ is odd}.
\end{cases}
$$
}
\State{\Return $h(h(f_l,f_c(t''-t-1),f_r)$}
\EndIf
\EndFunction
\end{algorithmic}
\end{algorithm}

\bigskip\begin{theorem} \label{thm:general}
Let $f$ be a composable function from strings to a set of size $n^{\bo(1)}$, computable in time $\time{f}$ for strings of length $1$, with its composing function $h$ being computable in time $\time{h}$. Then, given an RLSLP $G$ representing $T[1 \dd n]$, there is a data structure of size $\bo(|G|)$ that can be built in time $\bo(|G|(\time{f}+\time{h}\log n))$ and that computes any $f(T[i \dd j])$ in time $\bo(\time{h} \log n)$.
\end{theorem}

\begin{proof}
By Theorem~\ref{thm:balancing}, we can assume $G$ is balanced. We store the values $L[A] = |\exp{A}|$ and $F[A] = f(\exp{A})$ for every variable $A$, as arrays. These arrays add only $\bo(|G|)$ extra space because the values in $F$ fit in $\bo(\log n)$-bit words. Let us overload the notation and use $f(A,i,j)=f(\exp{A}[i\dd j])$. Algorithm~\ref{alg:general} shows how to compute any $f(A,i,j)$; by calling it on  the start symbol $S$ of $G$ we can compute $f(T[i\dd j]) = f(S,i,j)$.

Just as for ISLPs, in the beginning we follow a single path along the derivation tree, with only one recursive call per argument $A$ (lines 6, 8, and 17). The cost of those calls adds up to the height of the grammar, $\bo(\log n)$. This path finishes at a leaf or at an internal node $A$ where $\exp{A}[i\dd j]$ spans more than one child of $A$ in the derivation tree, in which case we may perform two recursive calls. 

Note that in the only places where this may occur (lines 10--11 and 18--19) those recursive calls will be prefix/suffix calls (i.e., either $i=1$ or $j=L[A]$ when we call \Call{f}{$A,i,j$}).

We define $c(A)$ as the highest cost to compute $f(A,i,L[A])$ or $f(A,1,j)$ over any $i$ and $j$ (i.e., the cost of prefix/suffix calls), charging $1$ to the number of calls to function \textsc{f} and $\time{h}$ to each invocation to function $h$. We assume for simplicity that $\time{h} \ge 1$ and prove by induction that $c(A) \le (1+2\time{h})d(A)+2\time{h}\log |\exp{A}|$, where $d(A)$ is the distance from $A$ to its deepest descendant leaf in the derivation tree. This certainly holds in the base case of leaves, where $d(A)=1$; it is included in lines 2--3. 

In the inductive case of rules $A \rightarrow BC$ (lines 4--12), we note that there can be two calls to \textsc{f}, but in prefix/suffix calls one of those calls spans the whole symbol---line 10 in a prefix call or line 11 in a suffix call. Calls that span the whole symbol finish in line 3 and therefore cost just $1$. Therefore, we have $c(A) \le 1+\max(c(B),c(C))+\time{h}$,
which by induction is
\begin{eqnarray*}
    c(A) & \le & 1+\max(c(B),c(C))+\time{h} \\
         & \le & 1+\time{h}+\max((1+2\time{h})d(B)+2\time{h}\log|\exp{B}|,(1+2\time{h})d(C)+2\time{h}\log|\exp{C}|) \\
         & \le & (1+2\time{h}) (1+ \max(d(B),d(C))) + 2\time{h} \log\max(|\exp{B}|,|\exp{C}|)) \\
         & \le & (1+2\time{h})d(A)+2\time{h}\log|\exp{A}|). 
\end{eqnarray*}

In the inductive case of rules $A \rightarrow B^t$ (lines 13--21), a similar situation occurs in lines 18--19: only one of the two recursive calls is nontrivial. Therefore, it holds $c(A) \le 1+c(B)+2\time{h}\log t+2\time{h}$, where the term $2\time{h}\log t$ comes from the recursive procedure to compute $f_c(t''-t'-1)$ in line 20; the logarithm is in base 2. Because $t=|\exp{A}|/|\exp{B}|$, by induction we have 
\begin{eqnarray*}
    c(A) & \le & 1+c(B)+\time{h} (2+2\log (|\exp{A}|/|\exp{B}|)) \\
         & \le & 1+(1+2\time{h})d(B)+2\time{h}\log|\exp{B}|+2\time{h}(1+\log (|\exp{A}|/|\exp{B}|))\\ 
         & = & (1+2\time{h}) (1+d(B)) + 2\time{h}\log|\exp{A}| ~=~ (1+2\time{h}) d(A) + 2\time{h}\log|\exp{A}|. 
\end{eqnarray*}

Therefore, the procedure costs $c(A) = (1+2\time{h})d(A)+2\time{h}\log|\exp{A}| = \bo(\time{h} \cdot \log n)$ from the nonterminal $A$ where the single path splits into two.

Arrays $L$ and $F$ can be precomputed in time $\bo(|G|(\time{f}+\time{h}\log n))$ via a postorder traversal of the grammar tree. We compute $f$ for every distinct individual symbol and $h$ for each distinct nonterminal $A$, whose children have by then their $L$ and $F$ entries already computed. In the case of rules $A \rightarrow B^t$, the entry $F[A]$ can be computed in time $\bo(t_h\log t)$ with the same mechanism used in line 20 of Algorithm~\ref{alg:general}.
\end{proof}
 
We show in the next section how to use this result to compute a more complicated function, which in particular we do not know how to compute efficiently on ISLPs.

\subsection{Application: Karp-Rabin Fingerprints}\label{sec:fingerprints}

Given a string $T[1 \dd n]$, a suitable integer $c$, and a prime number $\mu \in \bo(n)$, the Karp-Rabin fingerprint \cite{kr:ibm87} of $T[i \dd j]$, for $1 \le i \le j \le n$, is defined as 

$$\kappa(T[i \dd j]) ~=~ \Bigg(\sum_{k=i}^{j} T[k] \cdot c^{k - i}\Bigg)\bmod \mu.$$ 

Computation of fingerprints of text substrings from their grammar representation is a key component of various compressed text indexing schemes \cite{CEKNP20}. While it is known how to compute it in $\bo(\log n)$ time using $\bo(|G|)$ space on an RLSLP of size $G$ \cite[App.~A]{CEKNP20}, we show now a much simpler procedure that is an application of Theorem~\ref{thm:general}.

Note that, for any split position $p \in [i \dd j-1]$, it holds
\begin{equation}
\kappa(T[i \dd j]) ~=~ \bigg(\kappa(T[i \dd p]) + \kappa(T[p + 1 \dd j])\cdot c^{p - i + 1}\bigg)\bmod \mu.  \label{eq:partition}
\end{equation}
We use this property as a basis for the efficient computation of fingerprints on RLSLPs. 

\bigskip\begin{theorem}[cf.~\cite{fingerprints,CEKNP20}]
\label{thm:KR}
Given an RLSLP $G$ representing $T[1\dd n]$ and a Karp-Rabin fingerprint function $\kappa$, there is a data structure of size $\bo(|G|)$ that can be built in time $\bo(|G|\log n)$ and computes fingerprints of arbitrary substrings of $T$ in $\bo(\log n)$ time.
\end{theorem}
\begin{proof}
Let $f(X) = \langle \kappa(X),c^{|X|}\rangle$ be the function $f$ to apply Theorem~\ref{thm:general}. We then define
$$h(\langle \kappa_x,c_x\rangle,\langle \kappa_y,c_y\rangle) = 
     \langle (\kappa_x+\kappa_y \cdot c_x)\bmod \mu, (c_x \cdot c_y) \bmod \mu\rangle,$$
which can be computed in time $\time{h}=\bo(1)$.

It is easy to see that, by Eq.~(\ref{eq:partition}), $f(XY) = 
\langle \kappa(XY),c^{|XY|}\rangle = \langle (\kappa(X)+\kappa(Y)\cdot c^{|X|})\bmod \mu, (c^{|X|}\cdot c^{|Y|})\bmod\mu \rangle = h(\langle \kappa(X),c^{|X|}\rangle,
\langle \kappa(Y),c^{|Y|}\rangle) = h(f(X),f(Y))$. Therefore, application of Theorem~\ref{thm:general} leads to a procedure that computes $f(T[i\dd j]) = \langle \kappa(T[i\dd j]),c^{j-i+1} \bmod \mu\rangle$ in time $\bo(\log n)$ and using $\bo(|G|)$ extra space.
\end{proof}

\section{Conclusions}

We have generalized a recent result by Ganardi et al.~\cite{GJL2021}, which shows how to balance any SLP while maintaining its asymptotic size. Our generalization allows any rule of the form $A \rightarrow x$ where $x$ is a program, of size $|x|$, that generates the actual (possibly much longer) right-hand side. While we believe that this general result can be of wide interest to balance many kinds of generalizations of SLPs, we demonstrate its usefulness on a particular generalization of SLPs we call Iterated SLPs (ISLPs), which enable right-hand sides of the form $A \rightarrow \Pi_{i=k_1}^{k_2} B_1^{i^{c_1}}\cdots B_t^{i^{c_t}}$, of size $\bo(t)$. We say a grammar is a $d$-ISLP when $d$ is the maximum value of $c_j$ along all those rules; we also call $g_{it}$ the size of the smallest ISLP that generates a given text. 

ISLPs are interesting in the context of compressibility measures for repetitive texts. The well-known measure $\delta$ \cite{RRRS13,CEKNP20}, based on substring complexity and unreachable on some text families---i.e., their members cannot be encoded in $\bo(\delta)$ space---, had already been shown to be outperformed in some cases by a reachable measure, the size of an L-system generating the string \cite{NU21,NU23}. ISLPs are the first mechanism obtaining the same result---i.e., achieving size $\bo(\delta/\sqrt{n})$ on some texts of length $n$---while supporting polylogarithmic-time access---$\bo(\log^2 n \log\log n)$---to arbitrary text symbols. This result is obtained thanks to the possibility of balancing the ISLP, and is extended to computing other substring queries like range minima and next/previous smaller value, which are useful for implementing suffix trees on repetitive text collections \cite{GNP18}.

A further restriction, $d=\bo(1)$, yields $\bo(\log n)$ time for all the above queries, which is nearly optimal \cite{VY13} for accessing the text in any grammar-compressed form. This class includes Run-Length SLPs (RLSLPs), which extend SLPs with the rule $A \rightarrow B^t$ and are equivalent to $0$-ISLPs. We exploit again the possibility of balancing RLSLPs our result enables, and show a technique to efficiently compute a wide class of substring queries we call ``composable'', that is, where $f(X\cdot Y)$ can be computed from $f(X)$ and $f(Y)$. In particular, this enables us to compute Karp-Rabin signatures \cite{kr:ibm87} on the RLSLP-compressed string in time $\bo(\log n)$. This had been obtained before \cite[App.~A]{CEKNP20}, but with a much heavier machinery and requiring many more words of space per grammar symbol.

Several questions remain open on ISLPs. One is about the cost to find the smallest ISLPs that generates a given text $T$; we conjecture the problem is NP-hard as it is for plain SLPs \cite{CLLPPSS05} and RLSLPs \cite{KIKB24}. Indeed, since RLSLPs correspond exactly to $0$-ISLPs, finding the smallest $0$-ISLP is NP-hard. It is open to extend this result to $d$-ISLPs for other values of $d$.

A second question is whether we can build an index within $\bo(g_{it})$ space that offers efficient pattern matching. While ISLPs support random access to the text, the typical path followed for SLPs \cite{CNP21} and for RLSLPs \cite[App.~A]{CEKNP20} cannot be directly applied for ISLPs, because iteration rules, which are of size $\bo(t)$, would require indexing $\Theta(kt)$ positions. Computing Karp-Rabin fingerprints \cite{kr:ibm87} on text substrings, which can be done in logarithmic time on SLPs and RLSLPs and enable substring equality and longest common prefix computations on $T$, is also challenging on general ISLPs.

\backmatter

\section*{Declarations}

\subsection*{Funding}

All authors were funded with Basal Funds FB0001, ANID, Chile. F.O. was funded by scholarship ANID-Subdirección de Capital Humano/Doctorado Nacional/2021-21210579, ANID, Chile. C.U. was funded by scholarship ANID-Subdirección de Capital Humano/Doctorado Nacional/2021-21210580, ANID, Chile.

\subsection*{Competing interests}

All authors certify that they have no affiliations with or involvement in any organization or entity with any financial interest or non-financial interest in the subject matter or materials discussed in this manuscript.

\subsection*{Ethics approval}

Not applicable.

\subsection*{Consent to participate}

All authors agreed to participate in this work.

\subsection*{Consent for publication}

All authors agree for the publication of this work.

\subsection*{Availability of data and materials}

Not applicable.

\subsection*{Code availability}

Not applicable.

\subsection*{Authors contributions}

The three authors, G.N., F.O., and C.U. participated equally in the conception, proofs, and writing of the paper.

\bibliography{bibliography}% common bib file

\end{document}